\documentclass[letterpaper,11pt]{article}
\usepackage[utf8]{inputenc}
\usepackage[T1]{fontenc}
\usepackage{lmodern}
\usepackage{microtype}

\usepackage[left=1in,right=1in,top=1in,bottom=1in]{geometry} 

\usepackage{amsmath}
\usepackage{amsfonts}
\usepackage{amsthm}
\usepackage{mathtools}

\usepackage[procnumbered,ruled,vlined,linesnumbered,algo2e]{algorithm2e}

\usepackage{comment}

\usepackage[textwidth=1.5in,textsize=footnotesize,disable]{todonotes}

\usepackage[hidelinks]{hyperref}

\newtheorem{theorem}{Theorem}[section]
\newtheorem{lemma}[theorem]{Lemma}
\newtheorem{corollary}[theorem]{Corollary}
\newtheorem{claim}[theorem]{Claim}

\newcommand\givenbase[1][]{\:#1\lvert\:}
\let\given\givenbase

\DeclarePairedDelimiterX\Basics[1](){\let\given\sgiven #1}

\DontPrintSemicolon
\SetKwProg{Procedure}{Procedure}{}{}
\SetProcNameSty{textsc}
\SetFuncSty{textsc}
\SetKw{KwAnd}{and}
\SetKw{KwDo}{do}

\SetCommentSty{mycommfont}

\SetKwFunction{BartalTree}{BartalTree}
\SetKwFunction{Update}{Prune}
\SetKwFunction{AssignCenter}{AssignCenter}
\SetKwFunction{Initialize}{Initialize}
\SetKwFunction{Delete}{Delete}

\newcommand{\dem}{\operatorname{dem}}
\newcommand{\diam}{\operatorname{diam}}
\newcommand{\wdiam}{\operatorname{wdiam}}
\SetKwFunction{DeleteAuxiliary}{DeleteEdgeSet}
\newcommand{\id}{\operatorname{id}}
\SetKwFunction{UpdateClusterInformation}{UpdateClusterInformation}

\newcommand{\dist}{d}
\newcommand{\st}{\operatorname{st}}
\newcommand{\vol}{\operatorname{vol}}
\newcommand{\ball}{\operatorname{Ball}}
\newcommand{\cluster}{\operatorname{Clusters}}
\newcommand{\Ex}{\operatorname{Ex}}

\title{Dynamic Maintenance of Low-Stretch Probabilistic Tree Embeddings with Applications}
\author{Sebastian Forster\thanks{University of Salzburg, Austria. Supported by the Austrian Science Fund (FWF): P 32863-N.} \and Gramoz Goranci\thanks{University of Toronto, Canada. This work was partially done while at the University of Vienna. This research is partly supported by Natural Sciences and Engineering Research Council of Canada (NSERC). } \and Monika Henzinger\thanks{University of Vienna, Austria. The research leading to these results has received funding from the European Research Council under the European Union's Seventh Framework Programme (FP/2007-2013) / ERC Grant Agreement no. 340506.}
}
\date{}

\begin{document}

\maketitle

\thispagestyle{empty}

\begin{abstract}
We give the first non-trivial fully dynamic probabilistic tree embedding algorithm for a weighted graph $ G $ with $ n $ nodes and at most $ m $ edges undergoing edge insertions and deletions.
The goal in this problem is to maintain a tree containing all nodes of $ G $ with a randomized algorithm such that for every edge $ (u, v) $ of $ G $ the expected length of the path from $ u $ to $ v $ in the tree exceeds the weight of the edge $ (u, v) $ only by a small multiplicative factor, called the \emph{stretch} of the embedding.
In this paper, we obtain a trade-off between amortized update time and expected stretch against an oblivious adversary.
At the two extremes of this trade-off, we can maintain a tree of expected stretch $ O (\log^4 n) $ with update time $ m^{1/2 + o(1)} $ or a tree of expected stretch $ n^{o(1)} $ with update time $ n^{o(1)} $ (for edge weights polynomial in $ n $).
A guarantee of the latter type has so far only been known for maintaining tree embeddings with \emph{average} (instead of expected) stretch~[Chechik/Zhang, SODA '20].

Our main result has direct implications to fully dynamic approximate distance oracles and fully dynamic buy-at-bulk network design as our trade-off from above carries over to these two problems with minor overheads.
For dynamic distance oracles, our result is the first to break the $ O (\sqrt{m}) $ update-time barrier.
For buy-at-bulk network design, a problem which also in the static setting heavily relies on probabilistic tree embeddings, we give the first non-trivial dynamic algorithm.
As probabilistic tree embeddings are an important tool in static approximation algorithms, further applications of our result in dynamic approximation algorithms are conceivable.

From a technical perspective, we obtain our main result by first designing a decremental (i.e., deletions-only) algorithm for probabilistic low-diameter decompositions via a careful combination of Bartal's ball-growing approach [FOCS '96] with the pruning framework of Chechik and Zhang [SODA '20].
Such a low-diameter decomposition is the heart of Bartal's seminal tree embedding construction and we show how to adapt it to the decremental setting.
We then extend this to a fully dynamic algorithm by enriching a well-known ``decremental to fully dynamic'' reduction with a new \emph{bootstrapping} idea to recursively employ a fully dynamic algorithm instead of a static one in this reduction.
In contrast to previous applications of this type of reduction, such a bootstrapping can be applied efficiently because our decremental algorithm comes with the additional dynamic guarantee that each node changes its location in the tree only a logarithmic number of times over the course of the algorithm.

\end{abstract}

\newpage

\section{Introduction}
Approximating fundamental graph properties with a simpler graph has been an area of extensive research, leading to many powerful applications in the design of graph algorithms. A particularly simple graph of choice is a tree, since many graph problems admit somewhat easier solution on tree inputs. Typical examples include spanning trees (preserving connectivity), tree flow sparsifiers (preserving cuts and flows)~\cite{Racke02,HarrelsonHR03, RackeST14}, etc. The common goal behind these algorithmic tools is to ensure that the loss incurred when transferring to trees is as small as possible. 

One of the most powerful tree-based graph reductions has been the \emph{probabilistic tree embedding}~(PTE)~\cite{Bartal96,Bartal98,FakcharoenpholRT04} where the metric structure of an arbitrary graph is embedded probabilistically into special trees. More precisely, given a graph $G$, the goal is to find a probability distribution over a set $\tau$ of trees such that (1) distances in each tree of $\tau$ dominate those in $G$ and (2) expected distances in the random tree sampled from the distribution are within a factor of $\alpha$ from those in $G$, where $\alpha$ is usually referred to as \emph{stretch}. Building on the seminal work of Bartal~\cite{Bartal96,Bartal98}, Fakcharoenphol, Rao, and Talwar~\cite{FakcharoenpholRT04} showed that any graph $G$ admits a probabilistic tree embedding with stretch $O(\log n)$, which is existentially optimal for expander graphs. Their embedding can be computed in
time $O(m \log n)$~\cite{Blelloch0S17}.
This result has proven instrumental in designing approximation algorithms for a large class of problems including metric labeling~\cite{KleinbergT02}, buy-at-bulk network design~\cite{AwerbuchA97}, group steiner tree problem~\cite{GargKR00}, linear arrangement and spreading metrics~\cite{Bartal04}, vehicle routing~\cite{CharikarCGG98}, among others. 

Driven by the fundamental importance and applicability of probabilistic tree embeddings,  we initiate their study in the \emph{dynamic} setting. Concretely, the goal is to maintain a random tree $T$ with low stretch for a dynamic graph $G$ with $n$ vertices and maximum edge length $W$ that undergoes edge insertions and deletions such that after each update to $G$ the algorithm computes necessary changes to $T$. We seek to achieve a small time for handling edge updates while still being able to ensure that the expected   main result is the \emph{first} non-trivial  fully dynamic algorithm for maintaining a probabilistic tree embedding. 
Specifically, for edge weights polynomial in $ n $, we achieve $O(\log^4 n)$ stretch with $m^{1/2 + o(1)}$ time per operation, or  $n^{o(1)}$ stretch with $n^{o(1)}$ time per operation.
We actually show the following more general trade-off: Our algorithm achieves expected
stretch  $ (O (\log (n)))^{2 i - 1} (O (\log (nW)))^{i-1} $ using $ m^{1/i + o(1)} \cdot (O (\log (n W)))^{4i - 3} $ time per update for any integer $i \geq 2$. Note that for the related problem of dynamically maintaining spanning trees with low \emph{average} stretch the best result gives an average stretch 
of $n^{o(1)}$ in time $n^{o(1)}$~\cite{ChechikZ20}.

To demonstrate the applicability of our dynamic probabilistic tree embedding we show next that it enables novel dynamic algorithms for (1) distance oracles (aka dynamic APSP in a weighted, undirected graph) and (2) the buy-at-bulk network design problem. For the former, we present new trade-offs between the approximation factor and the running time guarantees of the oracle.
(1) Specifically
for any integer $i\ge 2$ we give a dynamic distance oracle with
an approximation ratio of $ (O (\log (n)))^{2 i - 1} (O (\log (nW)))^{i-1} $ and $ m^{1/i + o(1)} \cdot (O (\log (n W)))^{4i - 2} $ amortized update time. Note that \emph{no} fully dynamic distance oracle with $O(\sqrt{m})$ update time for any non-trivial approximation ratio was known before.

(2) In the \emph{buy-at-bulk network design} problem, we are given a weighted, undirected graph $G=(V,E,\ell)$,
where the length of each edge $e$ is $\ell_e$, and $ k $ source-sink pairs $ s_i, t_i $. Each $(s_i,t_i)$ pair has an associated demand  $\dem(i)$. Additionally, we are given a \emph{non-decreasing}, \emph{sub-additive} price function $f : \mathbb{R}_{\geq 0} \rightarrow \mathbb{R}_{\geq 0}$ that determines the cost $f(u)$ for purchasing a capacity $u$ on any edge in $G$. A feasible solution to the problem is a collection of paths $\{P_1,\ldots,P_k\}$ where $P_i$ is a path from $s_i$ to $t_i$ routing $\dem(i)$ units of commodity. Given a solution $\{P_1,\ldots,P_k\}$, let $c_e := \sum_{i:e \in P_i} \dem(i)$ be the total demand routed through the edge $e$. The goal is to \emph{find a feasible solution minimizing the total cost of the routing $\sum_{e \in E} \ell_e f(c_e)$.} We let $\mathrm{OPT}_G$ denote the total cost of the optimal solution. Note that the subbadditivity explicitly allows situations where doubling the capacity does not double the cost, which, for example, is the case with underground cables, where the cost is dominated by the cost of the excavation, and the cost of the cables placed underground is relatively small. 

We study a dynamic approximation algorithm for this problem that allows edge insertions and deletions to the input graph and supports queries of the form: Each query is given as parameter $k$ source-sink pairs $s_i,t_i$, each associated with a demand $\dem(i)$, and it returns an estimate that approximates $\mathrm{OPT}_G$. 
For this problem, we present
the \emph{first non-trivial} dynamic algorithm for the problem. Specifically for any integer $i\ge 2$ we give a fully dynamic
$ (O (\log (n)))^{2 i - 1} (O (\log (nW)))^{i-1} $-approximation algorithm with and $ m^{1/i + o(1)} \cdot (O (\log (n W)))^{4i - 3} $ amortized update time.

Taking cue from the usefulness of probabilistic tree embeddings, we believe that our techniques will find further applications in the future.

\paragraph{Technical contribution.} The main idea underlying our main result is to combine an iterative variant of Bartal's top-down construction in the static setting with a deletions-only algorithm that maintains \emph{probabilistic} low-diameter decompositions (LDDs).
It consists of three steps:

(1) For the decremental LDD algorithm we design a decremental version of Bartal's ball-growing process~\cite{Bartal96} that works against an oblivious adversary. This algorithm repeatedly picks an arbitrary vertex $c$ as center, selects a random radius (chosen from a suitable distribution) and removes the corresponding ball around $c$ from the graph. This decomposes the graph into balls.
Bartal then shows that for any edge the probability that it is an inter-ball edge is bounded.

We analyze a \emph{dynamic ball-growing process} where the above process is interleaved with arbitrary edge removals.
Intuitively, removing edges should not increase the probability that an edge  becomes an inter-cluster edge. However, a careful analysis is needed that
deals with dependencies that can arise as the ball-growing
process ``continues'' after each deletion step. 

To turn the balls grown by
this process into an LDD, we adapt the pruning  approach of Chechik and Zhang~\cite{ChechikZ20}.
That work associated a center with each cluster, we associate a center as well as a ball-growing process with each cluster.
Initially the whole graph is  one cluster, we pick a random center (according to a suitable distribution) and we initialize a ball-growing process for it. 
Now we repeatedly test for each cluster whether all its vertices are 
close to the center.
If
a vertex $v$ of a cluster lies too far from the cluster center, we use the ball-growing process of the cluster to grow a ball from $v$.
If the new ball has at most half the volume of its parent, it becomes a new cluster with center $v$
(removing its vertices from the old cluster, which might lead to the
removal of a set of edges from the ball-growing process of the old cluster) and we initialize
its own ball-growing process. Otherwise the ball around $v$ is \emph{not} turned into a cluster and instead we pick a new 
center for the \emph{old} cluster. Picking a new center and testing the distance of the vertices in the cluster from it is expensive as it takes time linear in its size, but
we show that whp the center of a cluster only changes $O(\log n)$ times.
This leads to a
hierarchy of clusters (and ball-growing processes) that is updated after each edge deletion. 
Using this approach, we can show that within total update time $ O(m^{1 + o(1)} \log W) $ (whp) we can maintain, for any given $\beta \in (0,1)$, a probabilistic LDD such that each cluster has weak diameter $O(\beta^{-1} \log^2 n)$ and each edge $e$ is an inter-cluster edge with probability at most $\beta w(e)$.

The main difference between this result and the decremental LDD algorithm of~\cite{ChechikZ20} is that (a) their work only bounds the \emph{total} number of
inter-cluster edges, while we bound \emph{for every edge} the probability that it is an inter-cluster edge\footnote{Note that if for each edge the probability of being an inter-cluster edge is at most $ \beta $, then the total number of inter-cluster edges is at most $ \beta m $ in expectation.} and (b)  their running time of $O(m \beta^{-1} \log^3n)$ (and,
thus, dependent on $\beta$), while ours is $O(m^{1 + o(1)} \log W)$
(and, thus, independent of $\beta$).
To achieve (a) we use Bartal's ball-growing process, to achieve (b) we use
the decremental 2-approximate SSSP algorithm of~\cite{HenzingerKN18} for maintaining estimates of the clusters' diameters, whose total update time is almost linear regardless of the cluster diameter parameter $\beta^{-1}$ (as opposed to an exact decremental algorithm like the Even-Shiloach tree~\cite{EvenS81,King99}).
Using the algorithm of~\cite{HenzingerKN18} adds certain complications because we use this algorithm to detect nodes that are too far away from the cluster center which are then removed from the cluster together with their balls.
If done naively, this approach would not respect the algorithm's oblivious adversary assumption.
We can circumvent this issue by \emph{not} reporting the removal of a ball from a cluster to the cluster center's instance of the decremental 2-approximate SSSP algorithm.
For this reason, the algorithm cannot faithfully detect when the diameter of a cluster becomes too large as it might still consider edges of removed balls for its distance estimates.
It can only detect if the distance of a node to its cluster center in a subgraph possibly larger than the cluster itself becomes too large. 
Thus, instead of providing a \emph{strong} diameter guarantee on the clusters, our probabilistic LDD only provides a \emph{weak} diameter guarantee; however, as demonstrated by Bartal~\cite{Bartal96}, the weak diameter guarantee is sufficient for constructing a PTE.

Note that the bottleneck of a running time depending on $\beta^{-1}$ was also inherent in the decremental probabilistic LDD algorithm of Forster and Goranci~\cite{ForsterG19} which implemented the random-shift clustering of Miller, Peng, and Xu~\cite{MillerPX13}. 
We would like to emphasize that the running time of our LDD algorithm being independent of the diameter parameter used for the cluster decomposition is key for making our PTE algorithm efficient.

(2) Equipped with this new decremental LDD algorithm we turn the hierarchical  static PTE algorithm of~\cite{Bartal96} into a decremental algorithm by maintaining \emph{one} decremental LDD algorithm per level in the hierarchy. 
This requires turning the
top-down approach of~\cite{Bartal96} into a bottom-up approach and
leads to a decremental PTE decomposition fulfilling the following four crucial properties: (a) The height of the resulting tree is $O(\log (nW))$, (b) 
the expected stretch is $O(\log^2 n \log (nW))$, (c)
the number of changes to the path from any vertex to the tree root is only polylogarithmic during the whole sequence of deletions, and (d) the total running time is
$O(m^{1 + o(1)} \log^2 W).$

(3) Based on this decremental PTE algorithm we then construct a fully dynamic PTE algorithm by using a ``recursive bootstrapping technique'':
A simple fully dynamic PTE algorithm can be achieved by running a static algorithm  after each update operations.
With the \emph{static} probabilistic tree embedding algorithm of~\cite{Blelloch0S17}, this gives a fully dynamic PTE algorithm that builds a tree with height $O(\log (nW))$, expected
stretch $O(\log n)$ and $O(m \log n)$ time per update. 

To bootstrap this result we proceed as follows:
We show how to turn (i) a decremental PTE algorithm with tree height $h_1$ and
stretch $s_1$ and (ii) \emph{a fully dynamic algorithm}  which outputs a tree of  height $h_2$ and
stretch $s_2$  into a faster fully dynamic algorithm which outputs a tree of height at most
$h_1 + h_2$, stretch $s_1 \cdot s_2$. The idea is 
that the decremental algorithm  outputs a compressed graph, called \emph{auxiliary graph},
on which the fully dynamic algorithm is run. The improvement in running time is achieved as the  auxiliary graph has size $O(k h_1)$, where $k$ is the number of updates since the beginning of the algorithm, resp., the last rebuild.

In step $ i $ of this bootstrapping scheme, the resulting tree height is $O(i \log (nW))$ and the stretch is
$ (O (\log (n)))^{2 i - 1} (O (\log (nW)))^{i-1} $. Furthermore by setting
$k =m^{1 - 1/i}$, (i) the  auxiliary graph has size $O(h_1
m^{1 - 1/i}) = \tilde O(m^{1 - 1/i})$ and (ii) the total running time of the decremental algorithm 
between two rebuilds can be amortized over $k$ operations, giving an amortized time of
$O(m^{1 + o(1)}/m^{1 - 1/i}) = O(m^{1/i + o(1)})$ per update operation.
On this auxiliary graph, the fully dynamic algorithm of step $i-1$ is executed, resulting in
a running time of $\tilde O((m^{1 - 1/i})^{1/(i-1) + o(1)}) = O(m^{1/i + o(1)})$ per update for the fully dynamic algorithms of step $i.$ This approach crucially depends on (1) the fact that the size of the auxiliary graph is only  $O(k h_1)$ and
(2) the fact that there are only polylogarithmically more updates in the fully dynamic algorithm than in the input graph. Fact (1) holds 
because  the decremental algorithm provides  a bound of $h_1$  on the height of its tree $T$ and each insertion of an edge $(u,v)$ leads to the insertion of the path from $u$, resp.~$v$ to the root of $T$ into the auxiliary graph. Fact (2) holds (a) because of the same reason as Fact (1) and (b) because the decremental algorithm
guarantees a polylogarithmic upper bound on the number of changes to the path from any vertex to the tree root and any such change leads to an update operation in the auxiliary graph.

\subsection*{Related Work}

 A topic closely related to our work is the dynamic maintenance of spanning trees with low average stretch.  Forster and Goranci~\cite{ForsterG19} were the first to study the problem and designed an algorithm with $n^{o(1)}$ average stretch and $n^{1/2+o(1)}$ update time per operation. Their result was subsequently improved by Chechik and Zhang~\cite{ChechikZ20}, who managed to keep the same stretch while bringing the update time down to $n^{o(1)}$. In contrast to our dynamic probabilistic tree embedding result, which ensures a low expected stretch for any vertex pair, the stretch guarantee in these works holds only on \emph{average}.

One of the most influential applications of probabilistic tree embeddings has been in the construction of distributions over trees that approximately preserves the cut and flow structure of a graph, known as \emph{cut-based decompositions}. In his groundbreaking work, R\"acke~\cite{Racke08} showed how to construct such a decomposition while losing only a logarithmic factor in the approximation, which is existentially optimal. Building upon the multiplicative weights update paradigm, his construction reveals that probabilistic tree embeddings~\cite{Bartal96,FakcharoenpholRT04} and cut-based decompositions are dual to each other. It has also led to the study of $j$-tree based graph approximation due to Madry~\cite{Madry10}, which has played a pivotal rule in the developments of approximating maximum flow in nearly linear time~\cite{KelnerLOS14, Sherman13, Peng16}.

Probabilistic embedding of graphs into tree metrics has been studied in other models of computation, including online~\cite{BartalFU20}, distributed~\cite{GhaffariL14, KhanKMPT08}, parallel~\cite{FriedrichsL18} and streaming algorithms~\cite{BeckerEL20}. The challenge of avoiding expensive \emph{exact} shortest path computations for construction probabilistic low-diameter decompositions has also been tackled by Becker, Emek and Lenzen~\cite{BeckerEL20}, who show that these computations can be replaced by a small number of approximate ones.
Building on the work of Bartal~\cite{Bartal96}, this technique implies algorithms for probabilistic tree embeddings in the CONGEST, PRAM and semi-streaming model, which are tight up to polylogarithmic factors. However, to the best our knowledge, their construction does not seem to extend to the dynamic setting.

There has been growing interest in maintaining graph-based decompositions or clustering, leading to breakthrough results for fundamental problems in dynamic graph algorithms. A prime example is the fully dynamic spanning tree algorithm of Nanongkai, Saranurak, and Wulff-Nilsen that crucially builds upon the dynamic maintenance of expander decompositions~\cite{NanongkaiS17, Wulff-Nilsen17, NanongkaiSW17, SaranurakW19}. Although expander decompositions readily imply low-diameter decompositions, this notion can only provide a guarantee on the total number of edges between clusters, which is not sufficient for our probabilistic decompositions. More importantly, the core idea underpinning their construction, the so called \emph{expander pruning} subroutine, has quadratic dependency on the expansion (and thus the diameter) parameter, which makes it inefficient for our purposes.  

The first sub-linear dynamic (approximate) distance oracle was developed by Abraham, Chechik and Talwar~\cite{AbrahamCT14}. For an unweighted, undirected graphs, they showed that there is a dynamic algorithm using $\tilde{O}(\sqrt{m}n^{1/k})$ expected amortized update time, $O(k^2\rho^2)$ query time and $2^{O(k\rho)}$ stretch, where $k \geq 2$ is an integer parameter and $\rho = 1+ \lceil\frac{\log n^{1-1/k}}{\log (m/n^{1-1/k})} \rceil$. While we believe that their result can be extended to weighted graphs, for example by using the near-optimal decremental distance oracle of Chechik~\cite{Chechik18}, prior to our work it was not clear how to circumvent their $\sqrt{m}$ barrier in the update time. Our new dynamic approximate distance oracle shows that this barrier can be overcome while keeping the stretch polylogarithmic.

\section{Preliminaries}
In the following we settle some basic notation and terminology and provide the most important definitions needed throughout the paper.

\paragraph{Basic Terminology.}

In this paper, we consider weighted, undirected graphs $ G = (V, E) $ with $ n $ nodes and (at most) $ m $ edges and with positive edge weights in the range from $ 1 $ to $ W $.
We refer to an edge with endpoints $ u $ and $ v $ as the (unordered) pair $ (u, v) $.
We denote the weight of an edge $ e \in E $ by $ w_G (e) $ or by $ w_G (u, v) $ if $ u $ and $ v $ are the endpoints of $ e $.
The \emph{degree} $ \deg_G (v) $ of a node is the number of nodes adjacent to $ v $, i.e., $ \deg_G (v) = | \{ u \in V \mid (u, v) \in E \} | $.
The \emph{volume} $ \vol_G (U) $ of a set of nodes $ U \subseteq V $ is the sum of the degrees of the nodes in $ U $, i.e., $ \vol_G (U) = \sum_{v \in U} \deg (v) $.
Note that $ \vol_G (V) = 2 |E| $.
We denote the sub-graph of $ G $ induced by a set of nodes $ U \subseteq V $ by $ G [U] $ (and $ G [U] = (U, E \cap U \times U) $).
For every pair of nodes $ u, v \in V $, the \emph{distance} $ \dist_G (u, v) $ between $ u $ and $ v $ is the length of the shortest path from $ u $ to $ v $ in $ G $.
The \emph{ball} around a node $ u $ with radius $ r $ is the set of nodes defined by $ \ball_G (u, r) = \{ v \in V \mid \dist_G (u, v) \leq r \} $.
The \emph{diameter} of~$ G $ is the maximum pairwise distance in $ G $, i.e., $ \diam (G) = \max_{u, v \in V} \dist_G (u, v) $.
The \emph{weak diameter}\footnote{In contrast, the \emph{strong diameter} of a set of nodes $ U $ is defined as $ \diam (G[U]) $, the diameter of the sub-graph induced by $ U $.} $ \wdiam_G (U) $ of a set of nodes $ U \subseteq V $ in $ G $ is the maximum distance in $ G $ between any pair of nodes of $ U $ , i.e., $ \wdiam_G (U) = \max_{u, v \in U} \dist_G (u, v) $.
Throughout the paper, we might omit the subscript indicating the graph we refer to if it is clear from the context.
We say that an event happens with high probability (whp) if it happens with probability at least $ 1 - \tfrac{1}{n^a} $ for any given constant $ a \geq 1 $.
We use the notation $ \tilde O (t) $ as an abbreviation for $ O (t \log^{O(1)} (nW)) $ (even if $ t $ does not depend on $ n $ or $ W $).

\paragraph{Dynamic Algorithms.}

In dynamic graph algorithms, the objective of the algorithm is to maintain some problem-specific output under updates to the input graph by spending as little time as possible after each update.
In this paper, we consider the following dynamic model:
An \emph{oblivious adversary} first fixes a (finite or infinite) sequence of graphs $ \mathcal{G} = G_0, G_1, \ldots $ where -- apart from the initial graph $ G_0 $ -- each $ G_i $ is obtained from $ G_{i-1} $ by applying an update operation.
In this paper, we consider two types of update operations, namely the insertion of a single edge or the deletion of a single edge, and assume that all graphs in the sequence share a common set of nodes.
After the adversary has chosen its sequence, the updates are revealed to the algorithm in an online fashion one at a time and after each update the algorithm must make all necessary changes to its output (to make it fit to the current graph in the sequence) before the next update is revealed.
For the sake of readability we usually do not specify the sequence $ \mathcal{G} $ explicitly and instead refer to the ``current'' version of the graph by the symbol $ G $ in the description and analysis of dynamic algorithms.

For \emph{fully dynamic} algorithms tolerating both edge insertions and edge deletions, we say that an algorithm has \emph{amortized update time} $ u (n, m) $ if for any $ k \geq 1 $ the total time spent for processing any sequence of $ k $ updates is at most $ k \cdot u (n, m) $ when starting from an empty graph with $ n $ nodes that during the sequence of updates has at most $ m $ edges (including the time needed to initialize the algorithm on an empty graph with $ n $ nodes before the first update).
For \emph{decremental} algorithms tolerating only edge deletions, we say that an algorithm has \emph{total update time} $ t (n, m) $ if the total time spent for any sequence of at most $ m $ deletions is at most $ t (n, m) $ when starting from a graph with $ n $ nodes and $ m $ edges (including the time needed to initialize the algorithm on a graph on $ n $ nodes and $ m $ edges before the first update).

\paragraph{Dynamic Tree Embeddings.}

A \emph{tree embedding} of a graph $ G = (V, E) $ is a forest $ T = (U, F) $ such that $ T $ contains the nodes of $ G $ and does not under-estimate the distances of $ G $, i.e., $ U \supseteq V $ and $ \dist_T (u, v) \geq \dist_G (u, v) $ for any pair of nodes $ u, v \in V $.
Note that for connected nodes $ u $ and $ v $ the path from $ u $ to $ v $ in $ T $ is unique.
For any pair of nodes $ u, v \in V $, the \emph{stretch} of $ u $ and $ v $ in $ T $ is the multiplicative factor by which their distance is over-estimated in $ T $, i.e., $ \st_T (u, v) = \tfrac{\dist_T (u, v)}{\dist_G (u, v)} $.
A \emph{probabilistic tree embedding} (PTE) of a graph $ G = (V, E) $ is a probability distribution $ \tau $ over tree embeddings of $ G $.
For any pair of nodes $ u, v \in V $, the \emph{expected stretch} of $ u $ and $ v $ in $ \tau $ is $ \st_{\tau} (u, v) = \Ex_{T \sim \tau} \st_T (u, v) $.
The \emph{maximum expected stretch} of $ \tau $ with respect to $ G $ is the maximum expected stretch of any pair of nodes, i.e., $ \st_{\tau} (G) = \max_{u, v \in V} \st_{\tau} (u, v) $.
Algorithmically, the goal is to design a randomized algorithm that computes a tree embedding $ T $ sampled from $ \tau $.
This distribution $ \tau $ might be defined implicitly by the random choices of such an algorithm.

A \emph{dynamic tree embedding} of a sequence of graphs $ \mathcal{G} = G_0, G_1, \ldots $ on the same set of nodes~$ V $ is a sequence of forests $ \mathcal{T} = T_0, T_1, \ldots $ such that each $ T_i $ is a tree embedding of $ G_i $.
The \emph{stretch} of a pair of nodes $ u, v \in V $ in $ \mathcal{T} $ at time~$ i $ is the stretch of $ u $ and $ v $ in $ T_i $, i.e., $ \st_{\mathcal{T}, i} (u, v) = \st_{T_i} (u, v) $.
A \emph{dynamic probabilistic tree embedding} (PTE) of a sequence of graphs $ \mathcal{G} = G_0, G_1, \ldots $ on the same set of nodes~$ V $ is a probability distribution $ \mathfrak{T} $ over \emph{dynamic tree embeddings} of $ \mathcal{G} $.
For any pair of nodes $ u, v \in V $, the \emph{expected stretch} of $ u $ and $ v $ in $ \mathfrak{T} $ at time $ i $ is $ \st_{\mathfrak{T}, i} (u, v) = \Ex_{\mathcal{T} \sim \mathfrak{T}} \st_{\mathcal{T}, i} (u, v) $.
The \emph{maximum expected stretch} of $ \mathfrak{T} $ with respect to $ \mathcal{G} $ is the maximum expected stretch of any pair of nodes at any time, i.e., $ \st_{\mathfrak{T}} (\mathcal{G}) = \max_i \max_{u, v \in V} \st_{\mathfrak{T}, i} (u, v) $.
We say that $ \mathfrak{T} $ is \emph{rooted} if for each dynamic tree embedding $ \mathcal{T} = T_0, T_1, \ldots $ of $ \mathfrak{T} $ each connected component in each $ T_i $ has a designated root and the \emph{height} of~$ \mathfrak{T} $ is the maximum number of edges on any root-to-leaf path.
Algorithmically, the goal is to design a randomized dynamic algorithm that maintains a dynamic tree embedding $ \mathcal{T} $ sampled from $ \mathfrak{T} $.
After each update to the graph, the dynamic algorithm needs to output the changes to the forest it maintains, i.e., all nodes and edges that are added to the forest or removed from it, respectively.

\paragraph{Dynamic Low-Diameter Decompositions.}

A \emph{clustering} $ C $ of a graph $ G = (V, E) $ is a partition of the nodes $ V $ into non-empty pairwise disjoint subsets called \emph{clusters} where for each node $ v \in V $ we denote the cluster of $ v $ by $ C (v) $. 
The weak diameter of $ C $ is the maximum weak diameter of any of its clusters, i.e., $ \wdiam_G (C) = \max_{v \in V} (\wdiam_G (C(v)) $.
A \emph{probabilistic weak $ (\beta, \delta) $-decomposition} (with $ \beta \in (0, 1) $ and $ \delta \geq 1 $) is a probability distribution $ \Gamma $ over clusterings of $ G $ such that each $ C \in \Gamma $ has weak diameter at most $ \delta $ and for every edge $ (u, v) \in E $ the probability of being an inter-cluster edge is at most a $ \beta $-fraction of its weight, i.e., $ \Pr_{C \sim \Gamma} [C (u) \neq C (v)] \leq \beta \cdot w_G (u, v) $.
If $ \delta $ is roughly proportional to $ \beta^{-1} $, we refer to such a decomposition as a \emph{probabilistic weak low-diameter decomposition (LDD)}.

A \emph{dynamic clustering} of a sequence of graphs $ \mathcal{G} = G_0, G_1, \ldots $ on the same set of nodes~$ V $ is a sequence of clusterings $ \mathcal{C} = C_0, C_1, \ldots $ such that each $ C_i $ is a clustering of $ G_i $.
The weak diameter of $ \mathcal{C} $ is the maximum weak diameter of any $ C_i $, i.e., $ \wdiam{\mathcal{G}} (\mathcal{C}) = \max_i \wdiam_{G_i} (C_i) $.
A \emph{dynamic probabilistic weak $ (\beta, \delta) $-decomposition} (with $ \beta \in (0, 1) $ and $ \delta \geq 1 $) of a sequence of graphs $ \mathcal{G} = G_0, G_1, \ldots $ on the same set of nodes~$ V $ with edge sets $ E_0, E_1, \ldots $ is a probability distribution $ \mathfrak{C} $ over dynamic clusterings of $ \mathcal{G} $ such that each $ \mathcal{C} \in \mathfrak{C} $ has weak diameter at most $ \delta $ and for every $ i $ and every edge $ (u, v) \in E_i $ the probability of being an inter-cluster edge at time $ i $ is at most a $ \beta $-fraction of its weight, i.e., $ \Pr_{\mathcal{C} \sim \mathfrak{C}} [C_i (u) \neq C_i (v)] \leq \beta \cdot w_{G_i} (u, v) $ (where $ C_i $ is the $i$-th clustering of the dynamic clustering $ \mathcal{C} $ sampled from $ \mathfrak{C} $).
If $ \delta $ is roughly proportional to $ \beta^{-1} $, we refer to such a dynamic decomposition as \emph{dynamic probabilistic weak low-diameter decomposition (LDD)}.

\section{Decremental Tree Embedding}
In this section, we develop an algorithm for maintaining a probabilistic tree embedding under deletions of edges.
We proceed as follows:
We first analyze a dynamic process for sequentially growing balls of randomly chosen radii, generalizing the original arguments of Bartal~\cite{Bartal96}.
We then design an algorithm for maintaining a probabilistic weak low-diameter decomposition under edge deletions by adapting an algorithmic idea of Chechik and Zhang~\cite{ChechikZ20} for maintaining such decompositions without the probabilistic guarantee.
Finally, we apply our decremental decomposition in an iterative manner to maintain a probabilistic tree embedding with expected-stretch guarantee on each non-tree edge under edge deletions in the input graph, following Bartal's original construction.

\subsection{Analysis of a Dynamic Ball-Growing Process}\label{sec:ball growing process}

Consider the following process on a weighted, undirected graph $ G = (V, E) $:
First, select a real $ p \in (0, 1) $, an integer $ k \geq 1 $, and $ k $ non-overlapping subsets $ E_1, \dots, E_k \subseteq E $ of edges (where set $E_i$ models
the edges deleted in the $i$-th ``round'').
Then repeat the following for $ i = 1 $ to $ k $ to grow balls $ B_1, \dots, B_k \subseteq V $:
\begin{enumerate}
\item Select an arbitrary vertex $ c_i $ of $ G_i = (V \setminus (B_1 \cup \dots \cup B_{i-1}), E \setminus (E_1 \cup \dots \cup E_{i-1})) $ (where $ G_1 = G $).
\item Randomly sample a value $ S_i $ from the geometric distribution\footnote{In this paper, the geometric distribution measures the number of Bernoulli trials needed to get the first success and thus $ R_i \geq 1 $.} with success parameter~$ p $ and grow a ball $ B_i $ from $ c_i $ in $ G_i $ of radius $ R_i = S_i - 1 $.

\end{enumerate}
Formally, we consider two adversaries in this process.
The first adversary chooses the next vertex $ c_i $ to grow a ball from and is \emph{fully adaptive} in the sense that it may see everything that has happened so far including the random choices made by our algorithm. 
The second adversary chooses the next set of edges to delete and is \emph{oblivious} to the random choices made by our algorithm. As the random
choices of the algorithm can only be revealed through the answers returned by the
algorithm, we require that the choices of the oblivious adversary
(i.e.~the sets $ E_1, \dots, E_k $) are fixed \emph{without knowledge of the
answers of the algorithm}. This is, for example, the case if they are fixed in advance, i.e., before the algorithm has started.
Note that for our application of this process in the next section the two adversaries as well as the fact  that each $E_i$ is a set of edges are important.

A static version of this process -- without edge deletions -- has been analyzed by Bartal~\cite{Bartal96} in his seminal work on tree embeddings.
Our analysis for the dynamic process follows a proof idea of Gupta~\cite{Gupta03} for the static process.

\begin{lemma}\label{lem:bound on radius of cluster}
For every $ 1 \leq i \leq k $, $ R_i \leq a p^{-1} \ln n $ with probability at least $ 1 - \tfrac{1}{n^a} $  for any $ a \geq 1 $.
\end{lemma}

\begin{proof}
Recall that for any $ k \geq 1 $ we have $ \Pr [S_i \geq k] = (1 - p)^{k-1}$,  the probability that a Bernoulli experiment with success probability $ p $ fails $ k-1 $ times in a row.
For $ k = \lfloor a p^{-1} \ln n \rfloor + 2 $ we get
\begin{equation*}
    \Pr [S_i \geq \lfloor a p^{-1} \ln n \rfloor + 2] = (1 - p)^{\lfloor a p^{-1} \ln n \rfloor + 1} \leq (1 - p)^{a p^{-1} \ln n} \leq e^{-a \ln n} = n^{-a} \, ,
\end{equation*}
where we employ the inequality $ (1 - p)^{p^{-1}} \leq e^{-1} $.
It now follows that
\begin{equation*}
\Pr [R_i \leq \lfloor a p^{-1} \ln n \rfloor] = \Pr [S_i \leq \lfloor a p^{-1} \ln n \rfloor + 1] = 1 - \Pr [S_i \geq \lfloor a p^{-1} \ln n \rfloor + 2] \geq 1 -n^{-a} \, .
\end{equation*}
\end{proof}

We say that an edge $ e $ is \emph{leaving a ball} if $e$'s endpoints are contained in two different balls $ B_i $ and $ B_j $ (with $ 1 \leq i, j \leq k $ and $ i \neq j $) or because one endpoint of $e$ is contained in some ball $ B_i $ (for some $ 1 \leq i \leq k $) and the other endpoint is contained in no ball at all, i.e., in $ V \setminus (B_1 \cup \dots \cup B_k) $.

\begin{lemma}\label{lem:probability for leaving a ball}
For any edge $ e \in E \setminus (E_1 \cup \dots \cup E_k) $, the probability that $ e $ is leaving a ball is at most $ p \cdot w (e) $.
\end{lemma}

\begin{proof}
Consider an arbitrary edge $ e \in E \setminus (E_1 \cup \dots \cup E_k) $ and let $ L_e $ be the event that $ e $ is leaving a ball.
Further, let $ C_e $ be the event that at least one of $e$'s endpoint is contained in a ball, i.e., in $ B_1 \cup \dots \cup B_k $.
We first argue that $ \Pr [L_e] \leq \Pr [L_e \given C_e] $:
As event~$ L_e $ implies event~$ C_e $, we have $ \Pr [L_e] = \Pr [L_e \cap C_e] $.
Now by the definition of conditional probability we have $ \Pr [L_e \cap C_e] = \Pr [L_e \given C_e] \cdot \Pr [C_e] $.
Since $ \Pr [C_e] \leq 1 $, the inequality $ \Pr [L_e] \leq \Pr [L_e \given C_e] $ follows.
We may therefore condition on event $ C_e $ to bound the probability of event $ L_e $ and assume in the following that at least one endpoint of $ e $ is contained in a ball.
For the sake of readability we will in the following abuse notation and not write the conditioning on $ C_e $ anymore.

Let $ u $ be the endpoint of $e$ that has been added to a ball first\footnote{To make the ordering precise, we assume that balls are grown in a breadth-first search manner adding nodes by increasing distance.}, and let $ v $ be the other endpoint.
Let $ i $ denote the number of the iteration in which $ u $ was added to a ball, let $ B_i $ be the corresponding ball, and let $ c_i $ be its center.
It must be the case that $ \dist_{G_i} (c_i, u) \leq \dist_{G_i} (c_i, v) $ by the choice of $ u $ being the endpoint that was added to a ball first.
Now $ e $ is leaving a ball if and only if $ v $ is not contained in the same ball as $ u $, i.e., if $ v \notin B_i $.
Observe further that $ v $ is not contained in $ B_i $ if and only if its distance to the center $ c_i $ exceeds the sampled radius $ R_i $, i.e., if $ \dist_{G_i} (c_i, v) > R_i $.
We therefore have $ \Pr [L_e] =  \Pr [\dist_{G_i} (c_i, v) > R_i] $.\footnote{As explained above, we are omitting the conditioning on $ C_e $ here on purpose to enhance the readability of the following inequalities.}
We will bound the complementary probability $ \Pr [\dist_{G_i} (c_i, v) \leq R_i] $.

Let $ x (e) = \dist_{G_i} (c_i, v) - \dist_{G_i} (c_i, u) $ be the difference in distance to $ c_i $ between both endpoints of~$ e $.
Note that $ x (e) \geq 0 $ and $ x (e) \leq w (e) $.
By the memorylessness of the geometric distribution we have 
\begin{align*}
    \Pr [R_i \geq \dist_{G_i} (c_i, v) \given R_i \geq \dist_{G_i} (c_i, u)] &= \Pr [R_i \geq \dist_{G_i} (c_i, u) + x (e) \given R_i \geq \dist_{G_i} (c_i, u)] \\
    &= \Pr [R_i \geq x (e)] = \Pr [S_i \geq x (e) + 1] \, .
\end{align*}

Recall that in the geometric distribution for the first ``success'' to appear after at least $ x (e) + 1 $ trials, the first $ x (e) $ trials must have been unsuccessful, which at each trial happens independently with probability $ 1 - p $.
By additionally applying Bernoulli's inequality, we get
\begin{equation*}
    \Pr [S_i \geq x (e) + 1] = (1 - p)^{x(e)} \geq 1 - p x(e) \geq 1 - p \cdot w(e) \, .
\end{equation*}
It follows that
\begin{equation*}
     \Pr [L_e] = \Pr [\dist_{G_i} (c_i, v) > R_i] = 1 - \Pr [\dist_{G_i} (c_i, v) \leq R_i] \leq p \cdot w(e) \, .
\end{equation*}
\end{proof}

Note that our proof crucially relied on the assumption that the sequence of edge deletions $ E_1, \dots, E_k $ was fixed \emph{without knowledge of the values $ S_1, \dots, S_k $}, i.e., it relied on the adversary generating the sequence of edge deletions being oblivious.
If any set $ E_i $ of edges to delete were selected \emph{with knowledge of} $ S_i $, then the proof would fail because $ e $ -- from the point of view of the random sampling process -- would not be an ``arbitrary'' remaining edge anymore as its choice might depend on the value of $ S_i $.
As an extreme case, we could, for example, observe the current balls and delete all edges that are internal to a ball. This would lead to a situation where each edge is leaving a ball with ``probability'' $ 1 $.

\subsection{Decremental Probabilistic Low-Diameter Decomposition}

Our strategy for maintaining a probabilistic weak low-diameter decomposition (LDD) is the following. Using a parameter $p \in (0,1)$,
the basic idea is to form clusters from balls obtained by the ball-growing process from Section~\ref{sec:ball growing process}.
This gives us clusters such that each edge has probability $ p $ of being an inter-cluster edge (which happens if the edge leaves a ball) and initially each cluster as diameter $ \tilde O (p^{-1}) $.
To detect whether the diameter of any cluster grows beyond this value due to edge deletions, we run a decremental approximate SSSP algorithm from a randomly chosen cluster center.
This decremental algorithm is initialized on the sub-graph induced by the respective cluster.
Whenever some node of the cluster is too far away from its center, we fix the situation by growing a new ball around this node with a randomly chosen radius to form a new cluster, and (under certain conditions) removing that ball from the cluster it originated from.
By our analysis of the ball-growing process, each edge is leaving this new cluster with probability $ p $ and the new cluster has diameter $ \tilde O (p^{-1}) $.
Thus, for each cluster, we spawn its own  ball-growing process, leading to  a hierarchy of clusters as new clusters are formed from balls in the ball-growing process of their ``parent'' cluster. The cluster of each ball-growing process can be modified by (a) edge deletions due to deletions in $G$ given by the \emph{oblivious} adversary and (b) removal of the ball of a ``child'' cluster, whose center we
chose so as to make the decremental LDD algorithm fast, i.e., the centers are chosen by a \emph{fully adaptive} adversary. This is the reason why we analyzed the ball-growing process in the previous section with these two adversaries. Also, as there might be no ball removals between a sequence of edge removals, we needed to use sets $E_i$ of edges in that process.

We want to ensure that the hierarchy of clusters has small depth for two reasons.
First, the running time for maintaining the decremental approximate SSSP algorithms has to be paid for each level of the hierarchy.
Second, each edge has a separate probability of leaving a ball (and thus being an inter-cluster edge) for each ball-growing process it participates in along the levels of the hierarchy.
Thus, by the union bound, each level adds a value of $ p $ to the total probability of being an inter-cluster edge.
Now to keep the depth of this hierarchy at $ O (\log n) $, we modify an idea of~\cite{ChechikZ20} and enforce that each newly formed cluster has at most half the volume of the cluster it originated from.
If this condition is not met for any ball we have grown, we do not form a new cluster from it and instead re-assign the cluster center by sampling from the nodes of the cluster with probability proportional to their degrees. Re-assigning the center potentially leads to new balls being grown to form clusters, but once we are done with that, the desired properties of our hierarchical clustering are established again.
However, this re-assignment of the center is potentially expensive because we need to restart the cluster's decremental approximate SSSP algorithm.
We can show that, due to our strategy for sampling the center, each cluster center is re-assigned only $ O (\log n) $ times, which is still tolerable for keeping the probability of being an inter-cluster edge and the running time within the desired bounds.
To bound the probability of being an inter-cluster edge by $ \beta $, we simply choose a value of $ p $ that is by a logarithmic factor smaller than the target value $ \beta $.

From a technical perspective, we heavily exploit that for certifying that a cluster still has bounded radius it is sufficient to main \emph{approximate} instead of exact distances from the cluster center.
This allows us to run for each cluster the decremental algorithm of Henzinger, Krinninger, and Nanongkai~\cite{HenzingerKN18}, whose total update time is almost linear regardless of the radius of the center maintained.
However, this algorithm assumes an oblivious adversary who is unaware of the previous answers given by the algorithm when generating the next edge deletion.
Since we still like to use the answers of this algorithm for our clustering decisions, we do \emph{not} report the removal of a ball from a cluster to the instance of the decremental approximate SSSP algorithm of that cluster's center.
We only report the deletion of edges from the input graph.
This means that the decremental approximate SSSP algorithm, although it was initialized on the sub-graph induced by the cluster, does not faithfully maintain a distance estimate within the sub-graph induced by the cluster.
However, it still gives good enough distance estimates between the center and nodes in the cluster with respect to the full graph.
This essentially means that the clusters maintained by our algorithm only have a \emph{weak} diameter guarantee instead of a strong one, which is sufficient for obtaining probabilistic tree embeddings.

The pseudocode for our approach is given in Algorithm~\ref{alg:decremental LDD}.
For every cluster $ C $ we use binary search trees to maintain an adjacency list representation of the following two sub-graphs: $ G [C] $, the sub-graph induced by $ C $, and $ H_C $, the sub-graph containing all edges present in $ G [C] $ since the last assignment of the cluster center except for those edges deleted in the meanwhile.
To keep the presentation succinct, we simply refer to ``forming a new cluster for a set of nodes'' as the act of initializing these data structures as well as creating a pointer for the cluster by which these data structures can be accessed.
Note that for every cluster $ C $, the decremental $2$-approximate SSSP algorithm $ \mathcal{A}_C $ is executed on $ H_C $.
Additionally, we maintain for every edge $ e $ the set of clusters $ \cluster (e) $ such that $ e $ is contained in $ H_C $ if and only if $ C \in \cluster (e) $
In the pseudocode, we slightly abuse the notation for the sake of readability by identifying clusters with their current set of nodes instead of explicitly using pointers.

\begin{algorithm2e}
\caption{Decremental Probabilistic Low-Diameter Decomposition}\label{alg:decremental LDD}

\Procedure{\Update{$ C $}}{     
    \If{$|C| > 1$ \KwAnd $\exists v \in C $ such that $ \delta (c, v) > 6 \rho $}{\label{line:if statement}
        Randomly sample a value $ S $ from the geometric distribution with success parameter~$ p $\;
        $ R \gets S - 1 $\;\label{line:sampling radius for new ball}
        $ B \gets \ball_{G[C]} (v, R) $\;\label{line:grow new ball} 
        \eIf{$ \vol (B) \leq \tfrac{1}{2} \mu_C $}{\label{line:case distinction small volume}
            $ C \gets C \setminus B $\;\label{line:shrink cluster}

            Form new cluster $ B $\;\label{line:form new cluster}
            \AssignCenter{$ B $}\;\label{line:recursive assign center}
            \Update{$ B $}\;\label{line:recursive update}    
        }{
            Stop algorithm $ \mathcal{A}_C $\;
            \lForEach{edge $ e \in F_C $}{
                $ \cluster (e) \gets \cluster (e) \setminus \{ C \} $
            }
            \AssignCenter{$ C $}\;
        }
        \Update{$ C $}\;
    }
}

\Procedure{\AssignCenter{$ C $}}{
    $ \mu_C \gets \vol_{G[C]} (C) $\;
    Assign sampling probability $ q_u = \tfrac{\deg_{G[C]} (u)}{\mu_C} $ to each node $ u $ and randomly sample a node $ c $ from this distribution\;\label{line:reassign center}
    $ F_C \gets E \cap C \times C $\;
    $ H_C \gets (C, F_C) $\;
    Initialize decremental $ 2 $-approximate SSSP algorithm $ \mathcal{A}_C $ with source $ c $ on $ H_C = (C, F_C) $ providing distance estimates $ \delta (c, \cdot) $\;\label{line:initialize decremental algorithm}
    \lForEach{edge $ e \in F_C $}{
        $ \cluster (e) \gets \cluster (e) \cup \{ C \} $
    }
}

\Procedure{\Initialize{$ G $, $ \beta $, $ a $}}{
    $ p \gets \tfrac{\beta}{2 + \log m} $\;
    $ \rho \gets (a+2) p^{-1} \ln n $\; 
    Form new cluster $ V $\;
    \AssignCenter{$ V $}\;
    \Update{$ V $}\;\label{update after initialization}
}

\Procedure{\Delete{$ e $}}{
    \ForEach{$ C \in \cluster (e) $}{
        $ \cluster (e) \gets \cluster (e) \setminus \{ C \} $\;
        $ F_C \gets F_C \setminus \{ e \} $\;
        Perform deletion of $ e $ in $ \mathcal{A}_C $\;
        \Update{$ C $}\;\label{update after deletion}
    }
}
\end{algorithm2e}

We start with the correctness proof by showing that the clusters maintained by our algorithm have the desired LDD properties.

\begin{lemma}\label{lem:decremental LDD bound on cut probability}
After the initialization and after processing each edge deletion in Algorithm~\ref{alg:decremental LDD}, the probability of being an inter-cluster edge is at most $ \beta w(e) $ for each edge.
\end{lemma}

\begin{proof}
Consider a tree with root $ V $ containing all clusters ever formed by the algorithm as nodes where a parent-child relationship between parent cluster $ C $ and child cluster $ B $ is established whenever the algorithm forms the new cluster~$ B $ during a call of \Update{$C$} in line~\ref{line:form new cluster} of Algorithm~\ref{alg:decremental LDD}.
As by the case distinction in line~\ref{line:case distinction small volume} the initial volume of clusters halves with each additional level in the tree (starting from initial volume $ 2 m $) and singleton-clusters have no children, the total number of levels in this tree is at most $ 2 + \log m $.

To each cluster $ C $ we assign a ball-growing process (see Section~\ref{sec:ball growing process}) for which the input graph is $ G [C] $ -- from which edges are deleted whenever they are also deleted from $ G $ -- and the balls grown by the process are those removed from the cluster $ C $ in line~\ref{line:shrink cluster} during calls of \Update{$C$} (which in our cluster tree are the child clusters of $ C $).
Recall that in the ball growing process the sequence of nodes chosen to grow balls from may be adapted to the random choices of the algorithm.
Thus, whenever the algorithm has grown a ball $ B $ from some vertex $ v $ and due to the volume constraint decides to not remove this ball, we consider these balls as not being part of the sequence of balls grown in the process.
We say that an edge \emph{participates} in the ball-growing process of a cluster $ C $ if it is contained in $ G [C] $ when $ C $ is formed.
Every edge participates in the ball-growing process of the root cluster $ V $.
Otherwise, an edge participates in the ball-growing process of a cluster only if it also participates in the ball-growing process of the parent cluster; this is the case because sibling clusters in the tree (which originate from the same ball growing process) are always pairwise (vertex) disjoint.
Therefore each edge participates in at most $ 2 + \log m $ ball-growing processes.

Considering all clusters formed and maintained by our algorithm, each edge can only be an inter-cluster edge if it is leaving a ball in one of the ball-growing processes it participates in.
The probability of an edge $ e $ leaving a ball in a single of these ball-growing process is at most $ p w (e) $ by Lemma~\ref{lem:probability for leaving a ball}.
By the union bound, the probability of an edge $ e $ leaving a ball in one of the at most $ 2 + \log m $ ball-growing processes it participates in is at most $ p w (e) (2 + \log m) = \beta w (e) $.
\end{proof}

We now provide the running time analysis of our algorithm.
We start with an auxiliary lemma that bounds the number of times cluster centers are re-assigned.

\begin{lemma}\label{lem:decremental LDD bound on diameter}
After the initialization and after processing an edge deletion using Algorithm~\ref{alg:decremental LDD}, every cluster has weak diameter at most $ 6 \rho = (6 (a + 2) (2 + \log m) \ln n) \beta^{-1} = O (a \beta^{-1} \log^2 n) $.
\end{lemma}

\begin{proof}
For every cluster $ C $, this is true whenever all calls to \Update{$ C $} are finished because by line~\ref{line:if statement} this procedure ensures that either the cluster has only size 1 and, thus, constant diameter or
$ \delta (c, v) \leq 6 \rho $ for the cluster center $ c $ and every node $ v \in C $.
The condition in the latter case implies $ \dist_G (c, v) \leq 6 \rho $ because the distance estimate $ \delta (c, v) $ maintained by $ \mathcal{A}_C $ (on the sub-graph $ H_C $ of $ G $) never underestimates the true distance, i.e., $ \delta (c, v) \geq \dist_{H_C} (c, v) \geq \dist_G (c, v) $.
Now observe that \Update{$ C $} is indeed called whenever a cluster is formed for the first time (by calling \AssignCenter{$ C $}) and whenever edges have been removed from the subgraph induced by $ C $.
\end{proof}

\begin{lemma}\label{lem:number of re-assignments}
With probability at least $ 1 - \tfrac{1}{n^a} $ for any given $ a \geq 1 $, 
for every cluster its center is re-assigned at most $ O (a \log n) $ times in Algorithm~\ref{alg:decremental LDD}.
\end{lemma}

Our proof is an adaptation of the corresponding proof in~\cite{ChechikZ20}.
We need to modify their proof for the following reasons:
(1) We pick the radius of a ball differently and consequently give,  for \emph{every} edge $e$, a bound on the probability that $e$ is an intercluster edge, while~\cite{ChechikZ20} gives a bound on the \emph{total number} of intercluster edges.
(2) We use the 2-approximate SSSP algorithm of~\cite{HenzingerKN18} instead of an Even-Shiloach tree~\cite{EvenS81} as used in~\cite{ChechikZ20}. These changes together with our new rule of forming balls  allows us to directly handle \emph{weighted} graphs,
 while \cite{ChechikZ20}  need additional techniques to handle weighted graphs.
 As we use an approximation instead of an exact algorithm, we need to show in our proof that the quality of the results is not affected (up to constant factors).
(3) More technically, we make the decision in Line 6 of Procedure \Update of whether a cluster is ``too big'' based on the volume of the cluster, while ~\cite{ChechikZ20} decides based on the number of \emph{internal} edges of a cluster.

\begin{proof}[Proof of Lemma~\ref{lem:number of re-assignments}]
In the following we work under the assumption that, every time the algorithm forms a new cluster in line~\ref{line:form new cluster} from a ball grown previously in line~\ref{line:grow new ball} with a radius $ R $ sampled from the geometric distribution with success parameter~$ p $ (see line~\ref{line:sampling radius for new ball}), we have $ R \leq \rho $.
By Lemma~\ref{lem:bound on radius of cluster}, this happens with probability at least $ 1 - \tfrac{1}{n^{a+2}} $ for each such random sampling.
As every cluster formed by the algorithm over the course of the algorithm always consists of at least one node, the algorithm forms at most $ n $~clusters in the worst case and thus our assumption holds with probability at least $ 1 - \tfrac{1}{n^{a + 1}} $ by the union bound.

Consider any cluster $ C $ for which the center~$c$ has just been re-assigned as by line~\ref{line:reassign center}.
Let $ C_0 $ denote the state of~$ C $ after this assignment and let $ G_0 $ denote the corresponding state of the input graph $ G $.
Let $ G_i = (V, E_i) $ be the status of the input graph $ G $ after the $i$-th subsequent edge deletion. Observe that $ E_0 \supset E_1 \supset \dots $.
The crucial definition for the remaining proof is the following:
Consider the last moment for which the sub-graph induced by $ C_0 $ contains a ball holding the majority of the initial volume.
Formally, let $ t $ be the largest index~$ i $ such that there exists a node $ v^* \in C_0 $ with $ \vol_{G_i [C_0]} (\ball_{G_i [C_0]} (v, \rho)) > \tfrac{1}{2} \mu_C $ and let $ B^* = \ball_{G_t [C_0]} (v^*, \rho) $.
Note that $ \vol_{G_0 [C_0]} (C_0) > \tfrac{1}{2} \mu_C $ because the ball $ B \subseteq C_0 $ grown directly before the re-assignment (line~\ref{line:grow new ball}) has volume more than $ \tfrac{1}{2} \mu_C $; therefore $ t $ is well-defined.
Furthermore $ B^* \subseteq C_0 $ by the definition of $ B^* $.

First, observe that the probability that the sampled node $ c $, to which the center is re-assigned, is contained in $ B^* $ is\footnote{Note that here we use the oblivious adversary assumption because the sequence $ G_0, G_1, \ldots $ -- and thus the set $ B^* $ -- has been chosen by the adversary prior to the random sampling of $ c $.}
\begin{align*}
\Pr [c \in B^*] =
\sum_{u \in B^*} q_u = \sum_{u \in B^*} \frac{\deg_{G_0 [C_0]} (u)}{\mu_C} 
 = \frac{\vol_{G_0 [C_0]} (B^*)}{\mu_C} 
 &\geq \frac{\vol_{G_t [C_0]} (B^*)}{\mu_C} \\
 &\geq \frac{\frac{1}{2} \mu_C}{\mu_C} = \frac{1}{2}
\end{align*}

For technical reasons, we now introduce a second indexing -- in addition to the indexing by number of deletions.
We define $ C^{(0)} = C_0 $ and let $ C^{(j)} $ denote the status of $ C $ at the beginning of the $j$-th call of \Update{$C$} after the re-assignment of the center.
Similarly, let $ G^{(j)} $ and $ H_C^{(j)} $ denote the status of $ G $ and $ H_C $, respectively, at that moment and let $ \delta^{(j)} (c, \cdot) $ denote the distance estimates produced by the decremental approximate SSSP algorithm with source $ c $ at that moment.
Observe that $ C^{(0)} \supseteq C^{(1)} \supseteq \dots $.
Finally, let $ \ell $ be the (largest) index $ j $ such that $ G^{(j)} = G_t $.

\underline{Claim 1:} If $ c \in B^* $, then $ B^* \subseteq C^{(j)} $ for every $ 0 \leq j \leq \ell+1 $.

Recall that $ B^* = \ball_{G_t [C_0]} (v^*, \rho) $ is the final ball of volume more than $ \tfrac{1}{2} \mu_C $ present in $ C_0 $.
Note further that we are still working under the initial assumption that the radii of the balls forming new clusters are bounded by $ \rho $.

We prove Claim~1 by induction.
The base case $ j = 0 $ holds by the definition of $ B^* $.
For the inductive step assume that $ B^* \subseteq C^{(j)} $ (where $ 0 \leq j \leq \ell $).
If $ C^{(j+1)} = C^{(j)} $, then clearly $ B^* \subseteq C^{(j+1)} $.
If $ C^{(j+1)} \subset C^{(j)} $, then a ball has been grown for some node $ v $ with $ \delta^{(j)} (c, v) > 6 \rho $ and then and removed from $ C^{(j)} $.
Let $ R \leq \rho $ be the sampled radius for this ball.
Assume for the sake of contradiction that $ \ball_{G^{(j)} [C^{(j)}]} (v, R) $ and $ B^* $ are not disjoint and contain some common node $ u $.
Then, by the triangle inequality, $ \dist_{G^{(j)} [C^{(j)}]} (c, v) \leq \dist_{G^{(j)} [C^{(j)}]} (c, v^*) + \dist_{G^{(j)} [C^{(j)}]} (v^*, u) + \dist_{G^{(j)} [C^{(j)}]} (u, v) $.
Since $ u \in \ball_{G^{(j)} [C^{(j)}]} (v, R) $, $ \dist_{G^{(j)} [C^{(j)}]} (u, v) \leq R \leq \rho $.
By the induction hypothesis we have $ B^* \subseteq C^{(j)} \subseteq C^{(0)} $.
As further $ B^* = \ball_{G_t [C^{(0)}]} (v^*, \rho) $ by definition, we have $ B^* = \ball_{G_t [C^{(j)}]} (v^*, \rho) $.
Since $ c \in B^* $ and $ u \in B^* $, $ \dist_{G_t [C^{(j)}]} (v^*, c) \leq \rho $ and $ \dist_{G_t [C^{(j)}]} (v^*, u) \leq \rho $.
Since $ j \leq \ell $, $ G_t $ is a sub-graph of $ G^{(j)} $, and therefore $ \dist_{G^{(j)} [C^{(j)}]} (v^*, c) \leq \dist_{G_t [C^{(j)}]} (v^*, c)  $ and $ \dist_{G^{(j)} [C^{(j)}]} (v^*, u) \leq \dist_{G_t [C^{(j)}]} (v^*, u) $.
It follows that $ \dist_{G^{(j)} [C^{(j)}]} (c, v) \leq 3 \rho $.
Now, since the approximate SSSP algorithm~$ \mathcal{A}_C $ provides a $2$-approximation (on the super-graph $ H_C^{(j)} $ of $ G^{(j)} [C^{(j)}] $), we get $ \delta^{(j)} (c, v) \leq 2 \dist_{H_C^{(j)}} (c, v) \leq 2 \dist_{G^{(j)} [C^{(j)}]} (c, v) \leq 6 \rho $, which contradicts $ \delta^{(j)} (c, v) > 6 \rho $.
Thus, $ \ball_{G^{(j)} [C^{(j)}]} (v, R) $ and $ B^* $ are disjoint, which implies the correctness of Claim~1.

\underline{Claim 2:} If $ c \in B^* $, then the center of $ C $ will never be re-assigned anymore (i.e., the else-branch in procedure \Update{$ C $} will not be executed anymore).

To prove Claim~2, it suffices to show that $ \vol_{G^{(j)} [C^{(j)}]} (\ball_{G^{(j)} [C^{(j)}]} (v, R)) < \tfrac{1}{2} \mu_C $ for every $ j \geq 0 $ such that there is a node $ v $ from which a ball of radius $ R $ is grown in line~\ref{line:grow new ball}.
If $ j \geq \ell + 1 $, then $ \vol_{G^{(j)} [C^{(j)}]} (\ball_{G^{(j)} [C^{(j)}]} (v, R)) \leq \vol_{G^{(j)} [C^{(0)}]} (\ball_{G^{(j)} [C^{(0)}]} (v, R)) = \vol_{G^{(j)} [C_0]} (\ball_{G^{(j)} [C_0]} (v, \rho)) < \tfrac{1}{2} \mu_C $, where the last inequality follows from the definitions of $ \ell $ and $ t $.
If $ j \leq \ell $, we argue as follows:
Any ball $ \ball_{G^{(j)} [C^{(j)}]} (v, R) $ grown in $ G^{(j)} [C^{(j)}] $ in line~\ref{line:grow new ball} and removed from $ C^{(j)} $ in line~\ref{line:shrink cluster} is disjoint from $ C^{(j+1)} $.
As $ B^* \subseteq C^{(j+1)} $ by the arguments above, then also $ B^* $ is disjoint from $ \ball_{G^{(j)} [C^{(j)}]} (v, R) $.
We therefore have 
\begin{equation*}
\vol_{G^{(j)} [C^{(j)}]} (B^*) + \vol_{G^{(j)} [C^{(j)}]} (\ball_{G^{(j)} [C^{(j)}]} (v, R)) \leq  \vol_{G^{(j)} [C^{(j)}]} (C^{(j)}) \leq \mu_C \, .
\end{equation*}
Since $ \vol_{G^{(j)} [C^{(j)}]} (B^*) > \tfrac{1}{2} \mu_C $, it follows that $ \vol (\ball_{G^{(j)} [C^{(j)}]} (v, R)) < \tfrac{1}{2} \mu_C $, which means that the center will not be re-assigned.
This completes the proof of Claim 2.

To finish the proof of the lemma, we view each re-assignment of the center as a Bernoulli trial with success probability at least $ \tfrac{1}{2} $, where we consider assigning the center to a node in the final majority-volume ball $ B^* $ as a success.
By standard arguments (see for example the proof of Lemma~\ref{lem:bound on radius of cluster}), the number of trials until the first success is $ O (a \log n) $ with probability at least $ 1 - \tfrac{1}{n^{a+2}} $.
This shows that for a single cluster the center changes $ O (a \log n) $ times with probability at least $ 1 - \tfrac{1}{n^{a+2}} $ if our initial assumption about the sampled radii being bounded by $ \rho $ holds.
The same guarantee holds for all clusters simultaneously with probability at least $ 1 - \tfrac{1}{n^{a+1}} $ as there are at most $ n $ clusters in the worst case over the course of the algorithm. 
Taking into account the small probability that our initial assumption might fail, we conclude that with probability at least $ 1 - \tfrac{1}{n^a} $, the cluster center changes $ O (a \log n) $ times for every cluster ever constructed by the algorithm.
\end{proof}

The second ingredient in the running time analysis is the observation that balls of bounded radius can be computed ``locally'' -- in time roughly proportional to the volume of the resulting ball.
\begin{lemma}\label{lem:local computation of balls}
There is a procedure that, given access to the adjacency list of a graph, computes the ball $ \ball (v, r) $ in time $ O (B \log B)$, where
$B = \vol (\ball (v, r)) $ for any given node $ v $ and radius $ r \geq 0 $.
\end{lemma}

\begin{proof}
We omit the proof which is a simple modification of Dijkstra's algorithm with binary heaps.
\end{proof}

\begin{theorem}\label{thm:decremental LDD}
Suppose we are given a decremental $ 2 $-approximate SSSP algorithm~$ \mathcal{A} $ with total update time $ t (m, n) $.
Then Algorithm~\ref{alg:decremental LDD} maintains, for any given $ \beta \in (0, 1) $ and $ \delta = (6 (a + 2) (2 + \log m) \ln n) \beta^{-1} = O(a \beta^{-1} \log^2 n) $, a probabilistic weak $ (\beta, \delta) $-decomposition of a weighted, undirected graph undergoing edge deletions 
such that with probability at least $ 1 - \tfrac{1}{n^a} $ each node changes its cluster $ O (a \log n) $ times over the course of the algorithm and the total update time is $ O (a t (m, n) \log^2 n + m \log^3 n) $ (for any given $ a \geq 1 $) and within this running time is able to report all nodes and incident edges of every cluster that is formed.
Over the course of the algorithm, each change to the partitioning of the nodes into clusters happens by splitting an existing cluster into two or several clusters and each node changes its cluster at most $ O (\log n) $ times.
\end{theorem}

\begin{proof}
The correctness of the algorithm has been established in Lemmas~\ref{lem:decremental LDD bound on cut probability} and~\ref{lem:decremental LDD bound on diameter}.
In our running time analysis we charge to each cluster $ C $ the time needed for performing all calls to \AssignCenter{$ C $} and \Update{$ C $} over the course of the algorithm excluding the induced calls to \AssignCenter{$ B $} and \Update{$ B $} for newly formed clusters in lines \ref{line:recursive assign center} and~\ref{line:recursive update} of Algorithm~\ref{alg:decremental LDD}, respectively.
Let $ n_C $ and $ m_C $ denote the number of nodes and edges of $ C $, respectively, when the cluster~$ C $ is formed.
We will argue that growing the balls in line~\ref{line:grow new ball} and running all instances of algorithm~$ \mathcal{A}_C $ (where each of them is initialized by executing line~\ref{line:initialize decremental algorithm}) takes total time $ O (a t (m_C, n_C) \log n_C + m_C \log^2 n_C) $.
All other operations are mere ``bookkeeping'' work, to for example maintain the nodes and edges of the graph $ G [C] $, that can be performed in total time $ O (m_C \log n_C) $.
By Lemma~\ref{lem:local computation of balls}, a ball $ \ball (v, R) $ of radius $ R $ around a node $ v $ in $ G [C] $ can be computed in time $ O
(V_C \log V_C)$, where $V_C = \vol_{G[C]} \ball_{G[C]} (v, R)$.
Whenever $V_C \leq \tfrac{1}{2} \mu $, the algorithm removes the nodes of $ \ball_{G[C]} (v, R) $ and the edges incident on $ \ball_{G[C]} (v, R) $ from $ G[C] $.
Therefore, by charging time $ O (\log m_C) $ to each edge initially contained in $ G [C] $, all of these computations of low-volume balls take time $ O (m_C \log n_C) $ in total.
Whenever $ \vol_{G[C]} \ball_{G[C]} (v, R) > \tfrac{1}{2} \mu $, we spend time $ O (m_C \log n_C) $ in the worst case and the algorithm calls \AssignCenter{$ C $} directly afterwards.
By Lemma~\ref{lem:number of re-assignments}, \AssignCenter{$ C $} is called $ O (a \log n_C) $ times with probability at least $ 1 - \tfrac{1}{n^a} $ for every cluster $ C $.
Therefore, all of these computations of high-volume balls take time $ O (a m_C \log^2 n_C) $ in total.
Finally, the decremental algorithm $ \mathcal{A}_C $ is restarted every time \AssignCenter{$ C $} is called.
Therefore the total time spent by algorithm $ \mathcal{A}_C $ is $ O (a t (m_C, n_C) \log n_C) $.
Summing it up gives that the total time charged to cluster $ C $ is $ O (a t (m_C, n_C) \log n_C + m_C \log^2 n_C) $.

Now consider a tree with root $ V $ containing all clusters ever formed by the algorithm as nodes where a parent-child relationship between parent cluster $ C $ and child cluster $ B $ is established whenever the algorithm forms the new cluster~$ B $ during a call of \Update{$C$} in line~\ref{line:form new cluster}.
Since clusters at the same level of the tree are (vertex) disjoint, the total time charged to all clusters at the same level is $ O (a t (m, n) \log n + m \log^2 n) $.
As by the case distinction in line~\ref{line:case distinction small volume} the initial volume of clusters halves with each additional level in the tree (starting from initial volume $ 2 m $) and singleton-clusters have no children, the total number of levels in this tree is at most $ O (\log n) $.
Therefore, the total time charged to all clusters is $ O (a t (m, n) \log^2 n + m \log^3 n) $.
As the remaining operations performed by the algorithm take time $ O (m) $, its total update time is $ O (a t (m, n) \log^2 n + m \log^3 n) $.
\end{proof}

Using the decremental $ (1 + \epsilon) $-approximate SSSP algorithm of Henzinger, Krinninger, and Nanongkai~\cite{HenzingerKN18}\footnote{For $ \epsilon = 1 $, Henzinger, Krinninger, and Nanongkai report a total update time of $ O (m^{1 + O (\log^{5/4} (\log n) / \log^{1/4} n)} \log W + n) $ in expectation. Using standard arguments, this can be turned into a high-probability bound at the expense of an additional logarithmic factor in the running time.} with $ \epsilon = 1 $, we arrive at a total update time that is independent of the diameter of the resulting clusters and thus does not depend on the value of the cut parameter $ \beta $.
\begin{corollary}\label{cor:decremental LDD}
There is a decremental algorithm to maintain, for any given $ \beta \in (0, 1) $ and $ \delta = (6 (a + 2) (2 + \log m) \ln n) \beta^{-1} = O (a \beta^{-1} \log^2 n) $ (where $ a \geq 1 $ is a given constant controlling the success probability), a probabilistic weak $ (\beta, \delta) $-decomposition of a weighted, undirected graph undergoing edge deletions that with high probability has total update time $ O (m^{1 + o(1)} \log W) $ and within this running time is able to report all nodes and incident edges of every cluster that is formed.
Over the course of the algorithm, each change to the partitioning of the nodes into clusters happens by splitting an existing cluster into two or several clusters and each node changes its cluster at most $ O (\log n) $ times.
\end{corollary}

\subsection{Decremental Probabilistic Tree Embedding}

In this section, we develop an algorithm for maintaining a probabilistic tree embedding under edge deletions. To simplify the notation we use the expression \emph{probabilistic weak low-diameter decomposition (LDD) with respect to $\Delta'$} to denote an LDD such that each cluster has diameter $\Delta'$.

Our construction is inspired by the static probabilistic tree embedding of Bartal~\cite{Bartal96}, which proceeds as follows: given a weighted, undirected graph $G$ and parameter $\Delta'/2$, where $\Delta' = \diam(G)$, it first computes a probabilistic weak LDD of $G$ with respect to $\Delta'/2$.
A rooted tree $T_i$ is recursively constructed in each $C_i$ with parameter $\Delta'/4$. Finally, a tree $T$ is output by creating an auxiliary root node $v_G$ and connecting it to the root nodes of all the $T_i$ trees, where the weight of each edge is set to $\Delta'$.

We use an iterative variant of the above algorithm with a fixed parameter $\Delta := nW$, where $W$ is the maximum edge length. Note that $\Delta$ is always an upper bound on $\diam(G).$ This particular choice of $\Delta$ is important as the diameter of $G$ can change after edge updates.
Another difference to Bartal's approach is that we use a single
LDD per level in the hierarchy, and not one LDD per cluster. 
The algorithm has $O(\log \Delta)$ iterations. For each vertex $v \in V(G)$ and iteration $i$, we record $\cluster(v,i)$, which is the cluster $v$ is assigned to at level $i$. Note that $\cluster(v,0) = 1$ for all $v \in G$ as $G_0 = G$ is the only cluster at level $0$. In iteration $i \geq 1$, given a graph $G_i := G_{i-1} \setminus \hat{E}^{(i-1)}$ and a parameter $\Delta/2^{i}$, we find a probabilistic LDD $C^{(i)}_1,\ldots,C^{(i)}_k$ of $G_i$ with respect to $\Delta/2^i$, where $\hat{E}^{(i-1)}$ is the set of inter-cluster edges from the LDD of $G_{i-1}$ and $\hat{E}^{(0)} = \emptyset$. For each $v \in C^{(i)}_j$ and $1 \leq j \leq k$, we set $\cluster(v,i) = j$. For each $C^{(i)}_j$ we do the following: we pick an arbitrary \emph{cluster representative} $v \in C^{(i)}_j$ to quickly find the ``parent cluster'' of $C^{(i)}_j$. Specifically, we connect $C^{(i)}_j$ and $C^{(i-1)}_{\cluster(v,i-1)}$ with weight $\Delta/2^{i-1}$, i.e., $C^{(i-1)}_{\cluster(v,i-1)}$ becomes the parent of $C^{(i)}_j$ at level $i-1$. This completes the description of an iteration. 

To summarize this construction maintains the following \emph{hierarchy invariant:}  
\begin{itemize} 
	\item All inter-cluster edges at level $i$ are deleted from the LDDs at levels $i+1,\ldots, \log_2 \Delta$.
\end{itemize}

For each graph $G_i$ we maintain one decremental LDD data structure $\mathcal{D}_{i}$ from the previous section. Recall that such a data structure requires the adversary to be oblivious to the previous answers of the data structure. The deletions that
$\mathcal{D}_{i}$ needs to 
process are: (1) edge deletions in $G$ given by the oblivious adversary and
(2) deletions of the edges in $\bigcup_{j=0}^{i-1}\hat{E}^{(j)}$. Note that the latter deletions depend only on the answers given by the decremental LDD data structures
$\mathcal{D}_{0}, \dots, \mathcal{D}_{i-1}$  of the graphs $G_0, \dots G_{i-1}$, and \emph{not} on the answers of $\mathcal{D}_{i}$. Thus, all these deletions are given by an adversary that is oblivious to the previous answers of $\mathcal{D}_{i}$.

To construct the probabilistic tree embedding we turn the hierarchy into a tree explicitly: We add for each cluster  on each level  of the hierarchy an  auxiliary node, except at the last level where we each singleton cluster consisting of a node $v$ is represented by $v$
itself.   

In a similar vein, we show how to iteratively use a decremental probabilistic weak LDD algorithm for maintaining a probabilistic tree embedding. Throughout, whenever we say that the decremental LDD from Corollary~\ref{cor:fullydyn} \emph{is initialized with a diameter parameter $\eta \geq 1$}, we mean to initialize it with $ \beta = (6 (a + 2) (2 + \log m) \ln n) \eta^{-1} $ (for some given constant $ a \geq 1 $ controlling the success probability) to make it guarantee a weak diameter of at most $ \eta $.
In contrast to the static algorithm, our construction of decremental LDDs will be implemented in a bottom-up approach. To this end,  let $G$ be the initial graph. Consider the hierarchy of decremental LDD data-structures $\mathcal{D}_{0}, \mathcal{D}_{1}, \ldots, \mathcal{D}_{\log_2 (\Delta) + 1}$ with diameter parameters $\Delta/2^{0}, \Delta/2^{1}, \ldots, \Delta/2^{\log_{2} (\Delta) + 1}$, respectively, created as follows. 
For levels $i = 0, 1, \ldots, \log_2 \Delta$, we set $C_{1}^{(i)} = G$, i.e., initially, for these levels we have only one cluster that corresponds to $G$ in $\mathcal{D}_i$. Observe that except at level $0$, these LDDs do not satisfy their diameter parameters and we will shortly see how to fix them. As each vertex is contained in only one cluster for these levels, we set $\cluster(v,i) = 1$, for each $v \in V(G)$ and $i = 0,1,\ldots, \log_2 \Delta$. At the last level of the hierarchy ($i = \log_2 (\Delta) + 1$), the diameter parameter is $\Delta/2^{\log_2 (\Delta) + 1} = 1/2$ and thus $\mathcal{D}_{\log_2 (\Delta) + 1}$ is simply the trivial, singleton clustering $\{v\}_{v \in G}$ of $G$. We introduce the first dependencies in the hierarchy by connection each $\{v\}$ at the last level to the single cluster $C_1^{\log_2 \Delta}$ with weight $1$. We also connect $C_1^{(i)}$ and $C_1^{(i+1)}$ with weight $\Delta/2^{i}$, for $i = 0,1,\ldots, \log_2(\Delta) - 1$.

We now make sure that the LDDs at levels $1,\ldots, \log_2{\Delta}$ satisfy their diameter parameters by proceeding in a bottom-up manner. Consider the penultimate level of the hierarchy, i.e., $i = \log_2 \Delta$. Using Corollary~\ref{cor:decremental LDD}, we initialize a decremental LDD $\mathcal{D}_{\log_2 \Delta}$ with diameter parameter $\Delta/2^{\log_2 \Delta} = 1$ of $C_1^{\log_2 \Delta} = G$.
Whenever a cluster $B$ splits off from a cluster $C$, we first update the cluster information and connections in the hierarchy involving vertices in $B$. Specifically, we start by choosing a vertex $u \in B$ as the representative of $B$ and connect $B$ to the cluster $C^{\log_2 \Delta}_{1}$ with weight $2$. This step ensures that $B$ has a parent cluster in the hierarchy. Assuming that each cluster at all levels in the hierarchy is assigned a unique id, we set $\cluster(v, \log_2 \Delta) = \id_B$ for each $v \in B$ and check whether $v$ is a representative of some cluster $C'$ at level $\log_2 (\Delta)+1$. If the latter holds, we remove the existing connection between $C'$ and $C$, and connect $C'$ to $B$ with weight $1$. This guarantees that the children clusters at level $\log_2 (\Delta)+1$, which were previously connected to $C$, now connect to $B$. 

Next, we proceed as before and initialize $\mathcal{D}_{\log_2 (\Delta) - 1}$ of $C_1^{\log_2 (\Delta)-1} = G$ with diameter parameter $\Delta/2^{\log_2 (\Delta) -1} = 2$, with the exception that whenever a cluster $B$ splits off from a cluster $C$ during the execution of $\mathcal{D}_{\log_2 (\Delta) - 1}$, in addition to updating the cluster information, we also take these new inter-cluster edges between $B$ and $C$ and pass them as deletions to $\mathcal{D}_{\log_2 (\Delta)}$. 
The last step ensures that hierarchy invariant holds. In general, during the initialization of $\mathcal{D}_{i}$ with parameter $\Delta/2^i$ at level $i$, whenever a cluster $B$ splits off from a cluster $C$, we update the cluster information and pass the new inter-cluster edges between $B$ and $C$ as deletions to data structures $\mathcal{D}_{i+1},\ldots,\mathcal{D}_{\log_2 \Delta}$. Observe that due to these deletions, new clusters might split off, which in turn can cause further deletions in the next levels.

The deletion of edges can be handled similarly: for any given edge $e$ to be deleted, we pass the deletion of $e$ to $\mathcal{D}_{i}$'s in a bottom-up approach, i.e.,~from $i = \log_2 \Delta$ to 1,  and recursively deal with the potential cluster split-offs that deletion of $e$ might trigger. The pseudocode of this construction is given in Algorithm~\ref{alg:decremental TreeEmbedding}.

\begin{algorithm2e}
\caption{Decremental Tree Embedding}\label{alg:decremental TreeEmbedding}

\SetKwFor{Whenever}{whenever}{do}{endw}

\Procedure{\Initialize{$G, a$}}{ 
       Let $\Delta := nW$ and 
        let $C_{1}^{(i)} = G$, for $i = 0,1,\ldots, \log_2 \Delta$ \;
        Set $\cluster(v,i) = 1$, for all $v \in V(G)$ and $i = 0,1,\ldots \log_2 \Delta$ \;
        Let $\{v\}_{v \in V(G)}$ be the singleton clusters of LDD $\mathcal{D}_{\log_2 (\Delta) + 1}$ with diameter parameter $1/2$ of $G$, and connect each $\{v\}$ to $C_{1}^{\log_2 \Delta}$ with weight $1$ \;
		Connect $C_1^{(i)}$ and $C_1^{(i+1)}$ with weight $\Delta/2^{i}$, for $i = 0,1,\ldots,\log_2(\Delta)-1$ 
        \For{$i = \log_2 \Delta,\ldots, 1$}
        { Invoke \Initialize{$C_1^{(i)},\Delta/2^{i},a$} in Algorithm~\ref{alg:decremental LDD} to get a decremental LDD $\mathcal{D}_{i}$ with diameter parameter $\Delta/2^{i}$ of $G$ \;
       	
	   \Whenever{a cluster $B$ splits off from a cluster $C$ in 	$\mathcal{D}_{i}$} 
		{   
		    \UpdateClusterInformation{$ B $, $C$} \;
         	\If{$i \leq \log_2 \Delta - 1$}
         	{ 				
				Let $E(B,C)$ denote the edges between $B$ and $C$ \;
				\For{$j = i+1,\ldots, \log_2 \Delta$}
			    	{
			    		\DeleteAuxiliary{$\mathcal{D}_{j}$, $E(B,C)$}
			   	    } 
		 	}
         }
        }
}

\Procedure{\Delete{$ e $}}{
	\For{$i = \log_2 \Delta, \ldots, 1$}
	{
    	\DeleteAuxiliary{$D_{i},\{e\}$}
    }
}

\Procedure{\DeleteAuxiliary{$\mathcal{D}_{i}$, $E$}}{
    \ForEach{edge $e \in E$}{
    Let $C_e$ be the cluster $e$ is assigned to in $\mathcal{D}_{i}$ \;
    Perform deletion of $e$ in $\mathcal{D}_{i}$ \;
    	\If{$C_e \neq \bot$}{
			\Whenever{a cluster $B$ splits off from a cluster $C$ in $\mathcal{D}_{i}$} 
			{ \UpdateClusterInformation{$B$, $C$} \;
		    \If{$i \leq \log_2 \Delta - 1$}
		    	{
					Let $E(B,C)$ denote the edges between $B$ and $C$ \;
					\DeleteAuxiliary{$\mathcal{D}_{i}$, $E(B,C)$}
				}
			}
    	}    
    }
}

\Procedure{\UpdateClusterInformation{$ B $, $C$}}{
Let $u \in B$ be the the chosen representative of $B$ \;
			  Connect $B$ and the cluster $C^{(i-1)}_{\cluster(u,i-1)}$ with weight $\Delta/2^{i-1}$ \;
			  \ForEach{$v \in B$}		
			  {	  
			  	Set $\cluster(v,i) = \id_B$ \;
			  	\If{$v$ is a representative of some cluster $C'$ at level $i+1$}
			  	{
			  		Delete the existing connection between $C'$ and $C$ \;
		    		Connect $C'$ and $B$ with weight $\Delta/2^{i}$ \;
			  	}
			  }
}
\end{algorithm2e}

Note that we maintain throughout the algorithm the property that every cluster has exactly one parent in the hierarchy. Thus the hierarchy induces a tree structure, which we denote by $T.$ 
Throughout, we let $G$ and $T$ refer to the current graph and induced tree, respectively. 

We start with the correctness proof by showing that the forest $T$ associated with our hierarchy of LDDs is a probabilistic tree embedding with the desired guarantees. A useful definition for our analysis is the following: we say that that two vertices $u$ and $v$ are \emph{separated} at level $i \geq 0$, if $u$ and $v$ belong to the same cluster at level $i$ but to different clusters at level $i+1$.

The following lemmata hold after initialization and after processing each edge deletion in Algorithm~\ref{alg:decremental TreeEmbedding}.

\begin{lemma} \label{lem:TEmbeddDomination} The following properties hold: (1) $V(G) \subseteq V(T)$ and (2) for all  $u,v \in V(G)$, $\dist_T(u,v) \geq \dist_G(u,v)$.
\end{lemma}
\begin{proof}
Property (1) follows immediately as there is a bijection between the leaf nodes of $T$ and the vertices of $G$. To show (2), suppose that $u$ and $v$ are separated at some level $i \geq 0$ and let $C^{(i)}$ the cluster they are belong to. By Corollary~\ref{cor:decremental LDD}, each cluster at level $i$ in $\mathcal{D}_i$ has weak diameter at most $\Delta/2^i$ and thus $\Delta/2^i \geq \dist_{G_i}(u,v)$. On the other hand, since $C^{(i)}$ is connected to its children clusters at level $i+1$ with weight $\Delta/2^{i}$, it follows that $\dist_T(u,v) \geq 2 \Delta/2^{i} \geq 2 \dist_{G_i} (u,v) \geq \dist_G(u,v)$. As for each pair of vertices there exists a level where they are separated, the lemma follows. 
\end{proof}

\begin{lemma} \label{lem:TEexpectedStretch}
For every $u,v \in V(G)$, $\Ex(\dist_T(u,v)) = O(\log^{2} n \log \Delta) \cdot \dist_G(u,v)$.
\end{lemma}
\begin{proof} It suffices to prove the lemma for every edge $e=(u,v) \in E$. We claim that if $u$ and $v$ are separated at level $i \geq 0$, then their distance in $T$ is at most $4 \Delta/2^{i}$. To see this, observe that the length of the path from a node at level $i$ in $T$ to some leaf node is at most $\sum_{j=i}^{\log_2 \Delta} \Delta/2^{j} \leq 2 \Delta/2^{i}$ and thus $\dist_T(u,v) \leq 4 \Delta/2^{i}$.

Let $A_i$ denote the event that the endpoints of an edge $e=(u,v)$ are separated at level $i \geq 0$. It follows that the expected stretch of $e$ in $T$ is
\[
	\Ex(\dist_T(u,v)) = \Pr[A_0] \cdot 4 \Delta + \sum_{i=1}^{\log_2 \Delta} \Pr[A_i \mid \bar{A}_{i-1} \cap \ldots \cap \bar{A}_{0}] \cdot \frac{4\Delta}{2^{i}}.
\] 

By applying Corollary~\ref{cor:decremental LDD} on the data-structure $\mathcal{D}_{i+1}$, the probability that $u$ and $v$ are separated at level $i$, conditioned on that they belong to the same cluster at level $i$, is bounded by  $O(2^{i+1} \log^{2} n / \Delta) \cdot w(e)$.  Thus each term in the above sum is of the form $O(2^{i+1} \log^{2} n / \Delta \cdot 4 \Delta/2^{i}) \cdot w(e) = O(\log^{2} n) \cdot w(e)$. Since there are at most $O(\log \Delta)$ terms, the claimed expected stretch on any edge follows.
\end{proof}

\begin{theorem} \label{thm:decremental TreeEmbedding}
There is a decremental algorithm to maintain a tree embedding of height $O(\log (nW))$ with expected stretch $O(\log^{2} n \log (nW))$ of a weighted, undirected graph $ G = (V, E) $ undergoing deletions that with high probability has total update time $ m^{1 + o(1)} \log^2 W $.
Over the course of the algorithm, for each vertex $v \in V$, the path from $v$ to its root vertex in the forest changes at most $O(\log n \log (nW))$ times.
\end{theorem}
\begin{proof}
We show that Algorithm~\ref{alg:decremental TreeEmbedding} has the claimed guarantees when using the decremental LDD algorithm of Corollary~\ref{cor:decremental LDD}.
Recall that $\Delta = nW$. The guarantees on the probabilistic tree embedding follow from Lemmas~\ref{lem:TEmbeddDomination} and~\ref{lem:TEexpectedStretch}. It remains to analyze the running time of Algorithm~\ref{alg:decremental TreeEmbedding} under a sequence of at most $m$ edge deletions. Note that total cost is dominated by (1) the time to initialize and maintain data structures $\mathcal{D}_1,\ldots,\mathcal{D}_{\log_2 \Delta}$ and (2) the time for updating the cluster information over the course of the algorithm. 

We claim that (1) is bounded by $m^{1 + o(1)} \log W \log \Delta $. To see this, consider the decremental data-structure $\mathcal{D}_i$ at level $i$. Over the course of the algorithm, we initialize $\mathcal{D}_i$ and then process edge deletions, each of which can be one of the following two types: deletions from the adversary in the current graph $G$ or deletions from the inter-cluster edges of LDDs $\mathcal{D}_1, \ldots, \mathcal{D}_{i-1}$. 
None of these deletions depends on the previous output of $ \mathcal{D}_{i}$ so that the oblivious adversary condition is fulfilled. 
By Corollary~\ref{cor:decremental LDD}, the initialization, together with all the deletions in $\mathcal{D}_i$ can be maintained in total time $ m^{1 + o(1)} \log W $. As there are $O(\log \Delta)$ many levels, the claimed bound follows. To analyze (2), note that whenever a cluster $B$ is split off from a cluster $C$, we can update the cluster information for $B$ and its dependencies in the hierarchy in $O(|B|)$ time. Hence, we can charge $O(1)$ to each vertex in $B$. 
By Corollary~\ref{cor:decremental LDD}, whp each vertex belongs to at most $O(\log n)$ newly formed cluster during the execution of $\mathcal{D}_i$. Thus, the running time per level in the hierarchy is at most $O(n \log n)$, which in turn implies that the total cost for maintaining this cluster information is bounded by $O(n \log n \log \Delta)$. 

We finally prove the bound on the number of path changes per vertex in the forest $T$ maintained by our algorithm. A leaf node in $T$~(or a vertex $v$ in $V$) changes its path to its root vertex in $T$ iff there exists an internal node in $T$~(corresponding to a cluster that contains $v$) that splits and $v$ is contained in the newly formed cluster. By Corollary~\ref{cor:decremental LDD}, over the course of the algorithm, $v$ can be contained in at most $O(\log n)$ newly formed clusters in some level $\mathcal{D}_i$. As there are $O(\log \Delta)$ levels, it follows that each vertex can change its path to its root vertex in $T$ at most $O(\log n \log \Delta)$ times.
\end{proof}

\section{Fully Dynamic Tree Embedding}

In this section, we generalize an approach of turning a decremental algorithm to a fully dynamic one, which has been used extensively in the context of dynamic algorithms for (approximate) shortest paths and reachability (see e.g.~\cite{HenzingerK95,RodittyZ11,AbrahamCT14}).

Roughly speaking, this approach runs a decremental algorithm on the graph and uses the information maintained by it to build a ``small'' problem-specific auxiliary graph whose nodes contain (representatives of) the endpoints of all edges inserted so far.
Then a static algorithm is executed on the auxiliary graph to compute the answer.
In this way, solutions for the fully dynamic setting are  ``stitched together'' from information from an algorithm that has only processed the deletions, but not the insertions. However, whenever the auxiliary graph gets ``too big'', the decremental algorithm is restarted and the auxiliary graph is rebuilt from scratch.
The decremental algorithm often only comes with a guarantee on its total update time for deleting all the edges -- regardless of the number of deletions actually performed. Employing amortized analysis, this total update time (including the initialization time) is charged to all updates occurring until the next restart of the decremental algorithm.

Our new idea is to ``bootstrap'' this construction by running a fully dynamic algorithm instead of a static algorithm on the auxiliary graph. Specifically we will end up with a hierarchy of fully dynamic algorithms: The bottom-most one runs the static algorithm on the auxiliary graph of size $ \tilde O (\sqrt{m}) $, resulting in the ``level-1'' fully dynamic algorithm with $O(m^{1/2 + o(1)})$ amortized update time. The ``level-$i$'' fully dynamic algorithm runs the ``level-($i$-1)'' fully dynamic algorithm on an auxiliary graph of size $O(m^{1-1/(i+1)})$, resulting in $O(m^{1/(i+1) + o(1)})$ amortized update time. The difficulty in our approach is that it requires a bound on the total number of changes to the auxiliary graph during each phase.
This graph might not only change after insertions, when nodes and edges are added to it, but also after deletions when the information maintained by the decremental algorithm changes, which in turn might lead changes in the auxiliary graph.
Bounding the number of the latter type of changes is challenging because, as mentioned above, decremental algorithms often only come with a bound on the total update time which sometimes gives a rather loose bound on the  total number of changes to the auxiliary graph, defeating the whole approach. We give a much tighter bound on the number of changes in the auxiliary graph in Theorem~\ref{thm:fully dynamic tree embedding trade-of}, making our bootstrapping approach work.

In the following, we first formulate a schematic ``decremental to fully dynamic'' reduction for probabilistic tree embeddings and then plug in concrete running times to obtain a trade-off between expected stretch and amortized update time.

\begin{lemma}\label{lem:decremental to fully dynamic}
Suppose we are given a decremental algorithm~$ \mathcal{A} $ for maintaining a rooted tree embedding $T_\mathcal{A}$ of height  at most $ h_\mathcal{A} $ with expected stretch at most~$ s_\mathcal{A} $ in total update time $ t_\mathcal{A} (m, n) $ such that for each node $ v $ the path $ p_v $ to its root in $ T_\mathcal{A} $ changes at most $ \chi_\mathcal{A} $ times and a fully dynamic algorithm $ \mathcal{B} $ for maintaining a rooted tree embedding $ T_\mathcal{B} $ of height at most $ h_\mathcal{B} $ with expected stretch at most~$ s_\mathcal{B} $ in amortized update time $ u_\mathcal{B} (m, n) $.
Then, for any integer $ k \geq 1 $ there is a fully dynamic algorithm $ \mathcal{C} $ for maintaining a rooted tree embedding of height at most $ h_\mathcal{A} + h_\mathcal{B} $ with expected stretch at most $ s_\mathcal{A} s_\mathcal{B} $ and amortized update time $ O (\tfrac{t_\mathcal{A} (m, n) \log (n)}{k} + \chi_\mathcal{A} h_\mathcal{A} \cdot u_\mathcal{B} (k h_\mathcal{A}, k h_\mathcal{A}) + h_\mathcal{A} \log (n)) $ if at least $ k $ updates are performed.
\end{lemma}

\begin{proof}
The algorithm subdivides the sequence of updates it receives into phases of length~$ k $.
In the following, we explain the algorithm's behavior during a fixed phase, where $ F $ denotes the set of edges present in the graph at the beginning of the phase and $ E $ always refers to the current set of edges.

We first define several sets and graphs, that algorithm~$ \mathcal{C} $ maintains during the phase:
\begin{itemize}
\item Let $ I = E \setminus F $ be the set of edges inserted to the graph since the beginning of the phase without subsequently having been deleted.
\item Let $ U = \{ v \in V \mid \exists e \in I : v \in e \} $ be the set of endpoints of edges in $ I $.
\item Let $ D = F \setminus E $ be the set of initially present edges deleted since the beginning of the phase.
\item Let $ T_\mathcal{A} $ be a rooted tree embedding of the sub-graph consisting of the edges $ F \setminus D $.
\item Let $ P = \bigcup_{v \in U} p_v $ be the graph that for each node $ v \in U $ (i.e., each endpoint of an inserted edge) contains the (unique) path $ p_v $ from $ v $ to its root in $ T_\mathcal{A} $.
\item Let $ H = I \cup P $ be the auxiliary graph that consists of all inserted edges and all tree paths of endpoints of inserted edges.
\item Let $ T_\mathcal{B} $ be a rooted tree embedding of~$ H $.
\item Let $ T_\mathcal{C} = (T_\mathcal{A} \setminus P) \cup T_\mathcal{B} $, which will be the output of algorithm~$ \mathcal{C} $, be the result of replacing $ P $ with $ T_\mathcal{B} $ in $ T_\mathcal{A} $.\footnote{More precisely, $ T_\mathcal{C} $ is obtained from $ T_\mathcal{A} $ by first removing all edges of $ P $ from $ T_\mathcal{A} $, then removing all nodes that have no neighbors anymore and finally adding all nodes and edges of $ T_\mathcal{B} $. Note that $ T_\mathcal{C} $ contains all nodes of the auxiliary graph $ H $ and possibly some additional Steiner nodes. We can imagine to obtain $ T_\mathcal{C} $ by ``gluing'' certain sub-trees of $ T_\mathcal{A} $ to leafs in $ T_\mathcal{C} $.}
\end{itemize}

To maintain these sets and graphs, the algorithm proceeds as follows:
At the beginning of the phase, the decremental algorithm~$ \mathcal{A} $ is initialized with edge set~$ F $ to maintain~$ T_\mathcal{A} $, and additionally, for each set and graph defined above, the algorithm initializes a binary search tree.
The tree embedding $ T_\mathcal{B} $ of $ H $ is maintained with the fully algorithm~$ \mathcal{B} $ and each change to $ H $ due to the operations described in the following will be processed as an update by $ \mathcal{B} $.
Whenever an edge $ e = (u, v) $ is inserted to the graph, the algorithm adds $ e $ to $ I $, adds $ u $ and $ v $ to $ U $, adds the paths $ p_u $ and $ p_v $ to $ P $ and to $ H $, and adds the edge $ e $ to $ H $.
Whenever an edge $ e $ is deleted from the graph, the algorithm considers two cases:
If $ e \in I $, the algorithm first removes all nodes and edges of $ p_u $ and $ p_v $ from $ P $ that are no longer contained in $ \bigcup_{v' \in U \setminus \{u, v \} } p_{v'} $ and also applies these changes to $ H $, then removes $ e $ from $ I $, removes $ u $ and $ v $ from $ U $, and finally removes $ e $ from $ H $ unless $ e $ is still contained in $ P $.
If $ e \in F $, then the algorithm first forwards the deletion to $ \mathcal{A} $ and then, for every node $ v' \in U $ for which the path $ p_{v'} $ from $ v' $ to its root in $ T_\mathcal{A} $ has changed it removes all nodes and edges of $ p_{v'} $ from $ P $ that are not contained in $ \bigcup_{u' \in U \setminus v'} p_{u'} $ and inserts the new unique path $ p_{v'} $ to $ P $; these changes are also applied to $ H $.
The tree $ T_\mathcal{C} $ is updated after each change to $ T_\mathcal{A} $ or $ T_\mathcal{B} $, and each change to $ T_\mathcal{C} $ is reported as an output of algorithm $ \mathcal{C} $.

We now prove that $ T_\mathcal{C} $ is the desired rooted tree embedding.
We first argue for the sake of completeness that $ T_\mathcal{C} $ is indeed a forest. Suppose that $ T_\mathcal{C} $ contains a cycle~$ K $.
Then $ K $ must contain edges from both $ T_\mathcal{A} $ and $ T_\mathcal{B} $ because neither $ T_\mathcal{A} $ nor $ T_\mathcal{B} $ contain cycles on their own.
Let $ u $ and $ v $ be the endpoints of a maximal sub-path $ S $ of $ K $ containing only edges of $ T_\mathcal{A} \setminus T_\mathcal{B} $.
Then $ u $ and $ v $ are contained in both $ T_\mathcal{A} $ and $ T_\mathcal{B} $ and are thus contained in the auxiliary graph $ H $.
Now observe that $ S $ is the unique path from $ u $ to $ v $ in $ T_\mathcal{A} $ and therefore must contain at least one edge from $ p_u $ (the path from $ u $ to the root in $ T_\mathcal{A} $) or from $ p_v $ (the path from $ v $ to the root in $ T_\mathcal{A} $).
However, such edges are included in $ P $ and can therefore only exist in $ T_\mathcal{C} $ if they are contained in $ T_\mathcal{B} $, which contradicts the definition of $ S $.
This shows that $ T_\mathcal{C} $ constains no cycles and is thus a forest
Clearly, each root of $ T_\mathcal{B} $ can serve as a root of $ T_\mathcal{C} $ and the height of $ T_\mathcal{C} $ then is at most $ h_\mathcal{A} + h_\mathcal{B} $.

We can bound the expected stretch of any edge $ e = (u, v) $ of~$ G $ in~$ T_\mathcal{C} $ as follows.
Let $ p_e  = (f_1, \ldots, f_\ell) $ be the unique path (as a sequence of edges) from $ u $ to $ v $ in $ T_\mathcal{A} $.
For any $ 1 \leq i \leq \ell $, define a path $ p_i $ as follows:
If $ f_i \notin P $, then $ f_i $ is also contained in~$ T_\mathcal{C} $ and we define $ p_i = f_i $.
If $ f_i \in P $, then $ f_i $ is also contained in~$ H $ and we define $ p_i $ to be the unique path in~$ T_\mathcal{B} $ from one endpoint of~$ f_i $ to the other endpoint of~$ f_i $.
Finally, we let $ p' = (p_1, \ldots, p_\ell) $ be the concatenation of all these paths.
Observe that $ p' $ is a path in $ T_\mathcal{C} $.
Fixing the random choices of algorithm $ \mathcal{A} $, the expected weight of path $ p' $ (over the random choices of algorithm~$ \mathcal{B} $) is
\begin{align*}
\Ex [w (p')] = \Ex \left[ \sum_{1 \leq i \leq \ell} w (p_i) \right] = \sum_{1 \leq i \leq \ell} \Ex [w (p_i)] \leq \sum_{1 \leq i \leq \ell} s_\mathcal{B} w (f_i) = s_\mathcal{B} w (p_e) \, .
\end{align*}
Furthermore, the expected weight of path $ p_e $ over the random choices of algorithm~$ \mathcal{A} $ is $ \Ex [w (p_e)] \leq s_\mathcal{A} \cdot w (u, v) $.
It follows that each edge $ e $ of $ G $ has expected stretch at most $ s_\mathcal{A} s_\mathcal{B} $ in $ T_\mathcal{C} $.

The amortized update time of~$ \mathcal{C} $ can be bounded as follows.
(1) The decremental algorithm~$ \mathcal{A} $ spends total time $ t_\mathcal{A} (m, n) $ for handling the at most $ k $ deletions per phase.
Therefore, by charging time $ t_\mathcal{A} (m, n) / k $ to each update of the previous phase (or to the current phase in case of the first phase, which the algorithm always completes due to the existence of at least $ k $ updates), we can account for the total time spent by the decremental algorithm.
(2) To analyze the total time spent by processing updates with algorithm~$ \mathcal{B} $, we first need to bound the total number of changes to $ H $ during the phase.
Observe that with each of the at most $ k $ insertions we add at most two paths, each consisting of at most $ h_\mathcal{A} $ edges, to $ P $.
As for every node the path to the root changes during all deletions in the phase at most $ \chi_\mathcal{A} $ times, the total number of changes to $ P $ during a phase  is $ O (k \chi_\mathcal{A} h_\mathcal{A}) $.
Additionally, the number of changes to~$ I $ is at most~$ k $.
Therefore, the number of changes to $ H $ during the phase is at most $ O (k \chi_\mathcal{A} h_\mathcal{A}) $.
Since the size of~$ H $ is $ O (|U| h_\mathcal{A}) = O (k h_\mathcal{A}) $ (both in terms of number of nodes and number of edges), this gives a total time of $ O (k \chi_\mathcal{A} h_\mathcal{A} \cdot u_\mathcal{B} (k h_\mathcal{A}, k h_\mathcal{A})) $ for processing all updates to~$ T_\mathcal{B} $ of the current phase.
Charging this time to the $ k $ updates of the previous phase (or to the current phase in case of the first phase), the amortized spent by $ \mathcal{B} $ with each update to $ G $ is $ O (\chi_\mathcal{A} h_\mathcal{A} \cdot u_\mathcal{B} (k h_\mathcal{A}, k h_\mathcal{A})) $.
Finally, there is some ``bookkeeping'' work to be done for maintaining the tree $ T_\mathcal{C} $ as well as the binary search trees for the sets $ I $, $ U $, and $ D $ and the graphs $ P $ and $ H $.
This can be done by charging time $ O (h_\mathcal{A} \log (n)) $ to each update, time $ O (\log (n)) $ to each change in $ T_\mathcal{A} $, and time $ O (1) $ to each change in $ T_\mathcal{B} $.
Overall, $ \mathcal{C} $ therefore has an amortized update time $ O (\tfrac{t_\mathcal{A} (m, n) \log (n)}{k} + \chi_\mathcal{A} h_\mathcal{A} \cdot u_\mathcal{B} (k h_\mathcal{A}, k h_\mathcal{A}) + h_\mathcal{A} \log (n)) $.
\end{proof}

\begin{theorem}\label{thm:fully dynamic tree embedding trade-of}
For every integer $ i \geq 2 $, when started on an empty graph, there is a fully dynamic algorithm for maintaining a rooted tree embedding of height $ i \cdot O (\log (nW)) $ with expected stretch $ (O (\log (n)))^{2 i - 1} (O (\log (nW)))^{i-1} $ that with high probability has amortized update time $ m^{1/i + o(1)} \cdot (O (\log (n W)))^{4i - 3} $. 
\end{theorem}

\begin{proof}
Observe first that -- by a standard technique -- it suffices to give an algorithm for the setting where $ m $ is a known upper bound on the maximum number of edges: We start with the upper bound being a constant and whenever the number of edges exceeds our upper bound we double the upper bound and restart the whole algorithm; the time needed for re-inserting the current edges of the graph after such a restart can be charged to the previous $ m/4 $ updates.

We now give an inductive proof in which for technical reasons\footnote{The case $ i = 1 $ is not included in the statement of the theorem because we are not obtaining a new result in this case.} we prove the statement for all $ i \geq 1 $.
The base case $ i = 1 $ holds due to the static algorithm of Blelloch, Gu, and Sun~\cite{Blelloch0S17} for computing an FRT-tree embedding~\cite{FakcharoenpholRT04}, which provides height $ O (\log (nW)) $ and expected stretch $ O (\log (n)) $, in time $ O (m \log (n)) = m^{1+o(1)} O (\log (nW) $.
For the inductive step, consider $ i \geq 2 $.
By the induction hypothesis, there is a fully dynamic algorithm~$ \mathcal{B} $ maintaining a rooted tree embedding of height $ h_\mathcal{B} = (i-1) \cdot O (\log (nW)) $ with expected stretch $ s_\mathcal{B} = (O (\log (n)))^{2 (i-1) - 1} (O (\log (nW)))^{(i-1)-1} $ in amortized update time $ u_\mathcal{B} (m, n) = m^{1/(i-1) + o(1)} \cdot (O (\log (nW)))^{4 (i-1) - 3} $.
We need to specify a fully dynamic algorithm~$ \mathcal{C} $ that starts from an empty graph.
As long as the number of updates to the graph is less than $ k = m^{1 - 1/i} $ we simply run algorithm $ \mathcal{B} $.
As the graph only has $ k = m^{1 - 1/i} = m^{(i-1)/i} $ many edges in this initial phase, we get an amortized update time of 
\begin{equation*}
    \left( m^{1 - 1/i} \right)^{1/(i-1) + o(1)} \cdot (O (\log (nW))^{4 (i-1) - 3} \leq m^{1/i + o(1)} \cdot O (\log (nW))^{4 (i-1) - 3} \, .
\end{equation*}

Whenever the number of updates to the graph exceeds the bound $ k = m^{1 - 1/i} $, we switch to algorithm~$ \mathcal{C} $ which is obtained by applying Lemma~\ref{lem:decremental to fully dynamic} with the fully dynamic algorithm~$ \mathcal{B} $ and the decremental algorithm~$ \mathcal{A} $ of Theorem~\ref{thm:decremental TreeEmbedding} maintaining a rooted tree embedding of height $ h_\mathcal{A} = O (\log (n W)) $ with expected stretch $ s_\mathcal{A} = O (\log^2 (n) \log (nW)) $ such that with high probability for each node the path to the root in $ T_\mathcal{A} $ changes at most $ \chi_\mathcal{A} = O (\log(n) \log(nW)) $ times and the total update time is $ t_\mathcal{A} (m, n) = m^{1 + o(1)} \log^2 (W) $.
We then arrive at a fully dynamic algorithm~$ \mathcal{C} $ for maintaining a rooted tree embedding of height $ h_\mathcal{A} + h_\mathcal{B} = O (\log (nW)) + (i-1) \cdot O (\log (nW)) = i \cdot O (\log (nW)) $ with stretch $ s_\mathcal{A} s_\mathcal{B} = O (\log^2 (n) \log (nW)) \cdot (O (\log (n)))^{2 (i-1) - 1} (O (\log (nW)))^{(i-1)-1} = (O (\log (n)))^{2 i - 1} (O (\log (nW)))^{i-1} $.
Setting $ k = m^{1 - 1/i} $, the amortized update time of this algorithm (with high probability) is
\begin{align*}
    \frac{m^{1 + o(1)} \log^2 (W) \log (n)}{k} &+ O (\log (n) \log^2 (nW)) \cdot (O (k \log (nW)))^{1/(i-1) + o(1)} \cdot (O (\log (n)))^{4(i-1) - 3} \\
    &\quad+ O (\log (nW) \log (n)) \\
    &= m^{1/i  + o(1)} \cdot O (\log^3 (nW)) \\
    &\quad+ m^{1/i + o(1)} \cdot (O (\log (nW)))^{1/(i-1) + o(1)} \cdot (O (\log (nW)))^{4(i-1)} + O (\log^2 (nW)) \\
    &= m^{1/i  + o(1)} \cdot (O (\log (nW)))^{4i - 3} \qedhere
\end{align*}
\end{proof}

Our theorem gives a trade-off between expected stretch and update time.
The smallest expected stretch, namely $ O (\log^3 (n) \log (nW)) $ is obtained for $ i = 2 $ which gives an update time of $ m^{1/2 + o(1)} \log^5 (nW) $.
To minimize update time, we balance the two terms $ m^{1/i} $ and $ (O (\log (nW)))^{4i-3} $ by setting $ i = \lceil \tfrac{\sqrt{\log (n)}}{\log (\log (nW))} \rceil $.
For this choice, both the update time and the stretch are sub-polynomial as long as the edge weights are polynomial in $ n $.

\begin{corollary}\label{cor:fullydyn}
For graphs with edge weights that are polynomial in $ n $, there is a fully dynamic algorithm for maintaining a rooted tree embedding of height $O (\log^{3/2} (n)) $ with expected stretch $ n^{o(1)} $ that with high probability has amortized update time $ n^{o(1)} $.
\end{corollary}

Finally, note that instead of running our algorithm directly on the input graph, we can also run it on a sparse spanner of the input graph to obtain further running time improvements.
For example, a fully dynamic spanner algorithm of Forster and Goranci~\cite{ForsterG19} for unweighted graphs can maintain a spanner of size $ O (n \log (n)) $ with stretch $ O (\log n) $ in expected amortized update time $ O (\log^3 (n)) $.
By binning the edges into weight ranges of doubling size and unioning the resulting spanners we get an algorithm for maintaining a spanner of expected size $ O (n \log (n) \log (W)) $ with stretch $ O (\log n) $ in expected amortized update time $ O (\log^3 (n)) $.
Note that the update time also trivially bounds the number of changes performed to the spanner with each update to the graph. Spanner algorithms give worst-case stretch guarantees w.r.t. to the input graph, which directly carry over multiplicatively to the probabilistic tree embedding.
Therefore we obtain the following guarantees.

\begin{corollary}\label{cor:fully dynamic tree embedding sparsified trade-off}
For every integer $ i \geq 2 $, there is a fully dynamic algorithm for maintaining a rooted tree embedding of height $ i \cdot O (\log (nW)) $ with expected stretch $ (O (\log (n)))^{2 i} (O (\log (nW)))^{i-1} $ that has expected amortized update time $ n^{1/i + o(1)} \cdot (O (\log (nW)))^{4i + 2} $.
\end{corollary}

\section{Applications}

\subsection{Dynamic Distance Oracle}
In this section, we show that our fully dynamic algorithm for maintaining a probabilistic tree embedding almost directly leads to a dynamic approximate distance oracle with comparable guarantees.  

An \emph{approximate distance oracle} of \emph{stretch} $ \alpha \geq 1 $ of a graph $ G = (V, E) $ is a data structure supporting a query operation that, for any given pair of nodes $ u, v \in V $ returns a distance estimate $ \delta (u, v) $ that never under-estimates the actual distance and over-estimates it by a factor of at most $ \alpha $, i.e., $ \dist_G (u, v) \leq \delta (u, v) \leq \alpha \cdot \dist_G (u, v) $.
The \emph{query time} is a bound on the time needed to support each query operation.
A \emph{dynamic approximate distance oracle} is a fully dynamic algorithm for maintaining such an approximate distance oracle.

The main insight behind our algorithm is to maintain logarithmically many independent copies of the probabilistic tree embedding data structure and upon receiving a distance query, compute in each tree the distance between the queried vertex pair and return the smallest distance as an estimate. Formally, let $\mathcal{D}_1, \ldots, \mathcal{D}_{a \log_2 n}$ be the data structures that dynamically maintain the probabilistic tree embeddings $T_1,\ldots,T_{a \log_2 n}$, where each $\mathcal{D}_i$ is obtained by invoking Theorem~\ref{thm:fully dynamic tree embedding trade-of} on an initially empty graph $G$ and $a$ is a non-negative parameter. Whenever an edge is inserted or deleted from $G$, we simply pass this update to each $\mathcal{D}_i$. Upon receiving a query about the distance between any vertex pair $(u,v)$ in $G$, we compute $\dist_{T_i}(u,v)$ in each $T_ i$ and return $\min_{i} \{\dist_{T_i}(u,v)\}$ as an estimate.

We next argue about the correctness and running time of the above dynamic distance oracle construction.

\begin{theorem} \label{thm:dynamicDistanceOracle}
	For every integer $i \geq 2$, when starting on an empty graph, there is a dynamic approximate distance oracle with query time $ O(i \log n \log (nW) $ that with high probability has stretch $ (O (\log (n)))^{2 i - 1} (O (\log (nW)))^{i-1} $ and amortized update time $ m^{1/i + o(1)} \cdot (O (\log (n W)))^{4i - 2} $ or stretch $ (O (\log (n)))^{2 i} (O (\log (nW)))^{i-1} $ with high probability and expected amortized update time $ n^{1/i + o(1)} \cdot (O (\log (nW)))^{4i + 3} $.
\end{theorem}
\begin{proof}
	By Theorem~\ref{thm:fully dynamic tree embedding trade-of} it follows that each $T_i$ is a probabilistic tree embedding with respect to the current graph $G$. To show the stretch guarantee, we will prove that with high probability, $\min_{i} \{\dist_{T_i} (u,v)\}$ is an $ (O (\log (n)))^{2 i - 1} (O (\log (nW)))^{i-1} $-approximation to $\dist_G(u,v)$. To this end, fix an arbitrary vertex pair $(x,y)$. By definition of probabilistic tree embeddings, for each $T_i$ we have that (i) $\dist_{T_i} (x,y) \geq \dist_G(x,y)$ and (ii) $\Ex[\dist_{T_i}(x,y)] \leq \alpha \cdot \dist_G(x,y)$, where $\alpha = (O(\log n))^{3i-2}$. Therefore, by Markov inequality, $\Pr[\dist_{T_i}(x,y) \geq 2 \alpha] \leq 1/2$, and hence
	\begin{align*}
	\Pr \left[\min_{i \in \{1,\ldots, a \log_2 n\}}  \{\dist_{T_i} (x,y) \geq 2 \alpha\} \right] &= \Pr[\dist_{T_i}(x,y) \geq 2 \alpha,~i = 1,\ldots, a \log_2 n]  \\ 
	& = \prod_{i=1}^{a \log_2 n} \Pr[\dist_{T_i}(x,y) \geq 2 \alpha] \leq (1/2)^{a \log_2 n} = n^{-a}.  		
	\end{align*} 
	
	Applying a union bound over at most $n^2$ distinct vertex pairs, we get that with probability at least $1 - n^{2-a}$, $\min_i \{\dist_{T_i}(u,v)\}$ is a $2\alpha = (O (\log (n)))^{2 i - 1} (O (\log (nW)))^{i-1} $ approximation to $\dist_G(u,v)$. 
	
	We next analyse the running time. Observe that the amortized update time of $ m^{1/i + o(1)} \cdot (O (\log (n W)))^{4i - 2} $ directly follows from Theorem~\ref{thm:fully dynamic tree embedding trade-of} as we maintain $O(\log n)$ copies of the tree embedding data structure, each having $ m^{1/i + o(1)} \cdot (O (\log (n W)))^{4i - 3} $ amortized update time. For the query time, Theorem~\ref{thm:fully dynamic tree embedding trade-of} guarantees that at any time each tree embedding $T_i$ has height $O(i \log(nW)) $. The latter implies that for any queried vertex pair $(u,v)$, we can compute the distance between $u$ and $v$ in $T_i$ in time $O(i \log(nW))$. As our construction maintains $O(\log n)$ tree embeddings, it follows that the query time is $O(i \log n \log (nW))$.
	The second trade-off on stretch and update time is obtained by plugging in the guarantees of Corollary~\ref{cor:fully dynamic tree embedding sparsified trade-off} instead of those of Theorem~\ref{thm:fully dynamic tree embedding trade-of}.
\end{proof}

\subsection{Dynamic Buy-at-Bulk Network Design}
In this section we present our dynamic algorithm for the buy-at-bulk network design problem. Recall that we are given a weighted, undirected graph $G=(V,E,\ell)$,
where the length of each edge $e$ is $\ell_e$, and a  \emph{non-decreasing}, \emph{sub-additive} price function $f : \mathbb{R}_{\geq 0} \rightarrow \mathbb{R}_{\geq 0}$~(i.e., (1) if $x \leq y$ then $f(x) \leq f(y)$ and (2) for any $x,y$, $f(x +y) \leq f(x) + f(y)$) that determines the cost $f(u)$ for purchasing a capacity $u$ on any edge in $G$.
The graph changes dynamically through edge insertion and deletions. Each query is given as input 
$ k $ source-sink pairs $ s_i, t_i $, each with an associated demand  $\dem(i)$, and outputs a value that is an approximation of the value of the optimal solution.

To achieve this, we follow the ideas of Awerbuch and Azar~\cite{AwerbuchA97} and the analysis of Williamson and Shmoys~\cite{DBLP:books/daglib/0030297}. The main observation is that the buy-at-bulk network design problem is easy to solve when the input graph is a tree: for each source-sink pair $ s_i, t_i $ there is a unique path $T_{s_i,t_i}$ connecting them in $ T $. Let $c^T_e = \sum_{i : e \in T_{s_i,t_i}} \dem(i)$ be the induced capacity for each edge $e \in T$. Thus, the optimal solution must purchase a capacity $c^T_e$ on each edge $e \in T$. We let $\mathrm{OPT}_T = \sum_{e \in E(T)} \ell^T_e f(c^T_e)$ denote the total cost of the optimal solution in $T$, where $\ell^T_e$ is the length of $e$ in $T$.

To solve the problem on general graphs, we can use an algorithm that produces a probabilistic tree embedding $T$, solve the problem on $T$ and then translate it back to the original graph. Awerbuch and Azar~\cite{AwerbuchA97} proved that this leads to an $O(\log n)$-approximation algorithm for the buy-at-bulk network design problem. 

Following the same approach, our dynamic algorithm proceeds as follows: given an initially empty graph $G$, we dynamically maintain a probabilistic tree embedding $T$ with $V(T) \supseteq V(G)$ using Theorem~\ref{thm:fully dynamic tree embedding trade-of}. Whenever an edge is inserted or deleted from $G$, we simply update our dynamic tree embedding $T$ with respect to this edge update.  Upon receiving a query about the optimal total cost of routing the demands $ \dem(1),\ldots,\dem(k)$, for $k$ source-sink pairs $s_i,t_i$, we do the following:
\begin{enumerate}
	\item Let $V_k : = \bigcup_{i} \{s_i,t_i\}$ be the union over vertices involved in the source-sink pairs.
	\item Construct the subtree $T' := \bigcup_{u \in V_k} T_{s,r_T}$ that consists of all the paths from vertices in $V_k$ to the root $r_T$ of $T$.
	\item Compute the optimal solution $\mathrm{OPT}_{T'}$ on $T'$ using the static algorithm described above.
	\item Return $\mathrm{OPT}_{T'}$ on $T'$ as an estimate.
\end{enumerate}

We next argue about the correctness (following~\cite{DBLP:books/daglib/0030297}) and running time of the above construction.

\begin{theorem} \label{thm:dynamicBuyAtBulk}
	For every integer $ i \geq 2 $, when started on an empty graph, there is a fully dynamic algorithm for maintaining an estimate that, in expectation, approximates up to an $(O (\log (n)))^{2 i - 1} (O (\log (nW)))^{i-1}$ factor (respectively up to an $ (O (\log (n)))^{2 i} (O (\log (nW)))^{i-1} $  factor) the cost of the optimal solution to any buy-at-bulk network design problem with $k$ source-sink pairs under edge insertions and deletions 
	in time $ m^{1/i + o(1)} \cdot (O (\log (n W)))^{4i - 3} $ per update operation (respectively time $ n^{1/i + o(1)} \cdot (O (\log (nW)))^{4i + 2} $ per update operation) and $O(i k \log (nW))$ time per query. 
\end{theorem}

\begin{proof}
	By Theoreem~\ref{thm:fully dynamic tree embedding trade-of}, it follows that $T$ is a probabilistic tree embedding with respect to the current graph $G$. Recall that there is a one-to-one correspondence between the leaf vertices of $T$ and the vertices in $G$. To show the approximation guarantee of the estimate returned by the query operation, we will prove that $\Ex[\mathrm{OPT}_T'] \leq (O (\log (n)))^{2 i - 1} (O (\log (nW)))^{i-1} \cdot \mathrm{OPT}_G$. We do this in several steps. First, observe that $T'$ contains the (unique) shortest path between the source-sink pairs $s_i,t_i$ in $T$, and thus $\mathrm{OPT}_{T'}  = \mathrm{OPT}_{T} $. Next, suppose there is an optimal solution that uses paths $\{P^*_1,\ldots, P^*_k\}$ in $G$. Then the optimal solution in $G$ uses capacity $c^*_e = \sum_{i:e \in P^*_i} \dem(i)$ on each edge $e$ and its cost is $\mathrm{OPT}_G = \sum_{e \in E} \ell_e f(c^*_e)$. We can easily translate the optimal solution in $G$ to $T$: for each $e = (u,v) \in E$, route $c^*_e$ units of demand along the shortest path $T_{u,v}$ between $u$ and $v$ in $T$. Let $\mathrm{cost}_T(\mathrm{OPT}_G)$ denote the cost of this routing in $T$. The claim below shows that this cost is at least as large as $\mathrm{OPT}_{T'}$.
	
	\begin{claim}
		$\mathrm{cost}_T(\mathrm{OPT}_G) \geq \mathrm{OPT}_{T'}$.
	\end{claim}
	\begin{proof}
		For each edge $e \in E(T)$, the optimal solution in $T$ uses capacity $c^T_e$, which in turn corresponds to the demand of all source-sink pairs $s_i,t_i$ that cross the cut induced by removing the edge $e$ from $T$. Observe that the translation of the optimal solution in $G$ gives another solution in $T$ which routes the demand $\dem_i$ between each pair $s_i,t_i$. Therefore, this solution must use at least capacity $c^T_e$ on each edge $e$ in $T$. Since $f$ is a non-decreasing function, it follows that  \[ \mathrm{cost}_T(\mathrm{OPT}_G) \geq \sum_{e \in E(T)} \ell^T_e f(c^T_e) = \mathrm{OPT}_T = \mathrm{OPT}_{T'}. \qedhere \]
	\end{proof}
	
	Using the above claim, to prove our approximation guarantee, it suffices to show that, in expectation, $\mathrm{cost}_T(\mathrm{OPT}_G)$ is at most $ (O (\log (n)))^{2 i - 1} (O (\log (nW)))^{i-1} \cdot \mathrm{OPT}_G$. To this end, as we will shortly prove, we claim that $\mathrm{cost}_T(\mathrm{OPT}_G)$ is at most $\sum_{e=(u,v) \in E}\dist_T(u,v) f(c^*_e)$. Since $T$ is a probabilistic tree embedding of $G$ with stretch $\alpha = (O (\log (n)))^{2 i - 1} (O (\log (nW)))^{i-1} $, we get that
	\begin{align*}
	\Ex\Bigg[\sum_{e=(u,v) \in E}\dist_T(u,v) f(c^*_e)\Bigg] & \leq \alpha \sum_{e=(u,v ) \in E} \dist_G(u,v) f(c_e^*) \\
	& \leq \alpha \sum_{e=(u,v ) \in E} \ell_e f(c_e^*) = \alpha \mathrm{OPT}_G.  
	\end{align*}
	
	To prove the claim, using the sub-additivity of $f$, we can bound $\mathrm{cost}_T(\mathrm{OPT}_G)$ as follows
	\begin{align*}
	\sum_{e' \in E(T)} \ell^T_{e'} & f\Bigg(\sum_{e=(u,v) \in E: e' \in T_{u,v}} c^*_e \Bigg)  \leq \sum_{e' \in E(T)} \ell^T_{e'}  \sum_{e=(u,v) \in E: e' \in T_{u,v}} f(c^*_e) \\
	& = \sum_{e=(u,v) \in E} f(c^*_e) \sum_{e' \in T_{u,v}} \ell^T_{e'} = \sum_{e=(u,v) \in E} \dist_T(u,v) f(c^*_e), 
	\end{align*}
	which completes the correctness proof. 
	
	We next analyze the running time. Since we maintain a single tree embedding data-structure, by Theorem~\ref{thm:fully dynamic tree embedding trade-of} it follows that our construction has an amortized update time of $ m^{1/i + o(1)} \cdot (O (\log (n W)))^{4i - 3} $ per operation. For the query time, Theorem~\ref{thm:fully dynamic tree embedding trade-of} guarantees that any tree $T$ has height $O(i \log (nW))$. The latter guarantees that the length of each path from a leaf vertex  to the root in $T$ is $O(i \log (nW))$, which in turn implies that the size of the subtree $T'$ is bounded by $O(ik\log (nW))$. As it is easy to see that an optimal solution to the buy-at-bulk network design problem in $T'$ can be computed in time proportional to its size, we conclude that the query time is $O(ik\log (nW))$.
	
	The second trade-off on approximation ratio and update time is obtained by plugging in the guarantees of Corollary~\ref{cor:fully dynamic tree embedding sparsified trade-off} instead of those of Theorem~\ref{thm:fully dynamic tree embedding trade-of}.
\end{proof}
The approximation ratio can be turned into a high-probability bound by running $\Theta(\log n)$ copies of the data structure and returning the minimum of their values. This increases the running time by a $\Theta(\log n)$ factor.

\bibliography{references}
\bibliographystyle{alpha}

\end{document}